\documentclass[runningheads]{llncs}

\usepackage[all]{xy}
\usepackage[]{url}
\usepackage[utf8]{inputenc}
\usepackage[]{hyperref}
\usepackage[T1]{fontenc}
\usepackage[]{stmaryrd}
\usepackage[]{bm}
\usepackage[]{array}
\usepackage[]{amssymb}
\usepackage[]{amsmath}

\institute{Laboratory of Informatics of Grenoble (CNRS, Grenoble INP)\\
    Bâtiment IMAG C - 220 rue de la Chimie 38400 Saint Martin d’Hères\\
    \{vincent.aravantinos,ricardo.caferra,nicolas.peltier\}@imag.fr}\spnewtheorem{definition2}[theorem]{Definition} {\bfseries} {\itshape}\spnewtheorem{remark2}[theorem]{Remark} {\itshape} {\rmfamily}\spnewtheorem{example2}[theorem]{Example} {\itshape} {\rmfamily}\spnewtheorem{lemma2}[theorem]{Lemma} {\bfseries} {\itshape}\spnewtheorem{proposition2}[theorem]{Proposition} {\bfseries} {\itshape}\spnewtheorem{corollary2}[theorem]{Corollary} {\bfseries} {\itshape}\toctitle{Linear Temporal Logic and Propositional Schemata, Back and Forth}\titlerunning{Linear Temporal Logic and Propositional Schemata, Back and Forth}\authorrunning{Vincent Aravantinos, Ricardo Caferra, Nicolas Peltier}

\title{Linear Temporal Logic and Propositional Schemata, Back and Forth\thanks{This work has been partly funded by the project ASAP of the French \emph{Agence Nationale de la Recherche} (ANR-09-BLAN-04-07-01)} (extended version)}
\author{Vincent Aravantinos, Ricardo Caferra, Nicolas Peltier}

\begin{document}
  {\maketitle}
  \begin{abstract}
    This paper relates the well-known formalism of Linear Temporal Logic \cite{ltl}
      with the logic of \emph{propositional schemata} introduced in \cite{tab09}.
      We prove that LTL is equivalent to a particular class of schemata in the  sense that polynomial-time 
      translation algorithms exist from one logic to the other.
      Some consequences about complexity are given.
      We report about first experiments and the consequences about possible improvements in existing implementations are analyzed.
      \end{abstract}
  \section{Introduction}Linear Temporal Logic (LTL) is a very well-known logic introduced in \cite{ltl} for verifying computer programs.
      It is widely used to reason on finite state transition systems.
      On the other hand, \emph{propositional schemata} have been introduced in \cite{tab09}.
      They extend the language of propositional logic with \emph{indexed propositions} (such as $p_{\mathsf{n}}$, $p_{1}$ or $p_{\mathsf{i}+1}$)
      and \emph{iterated connectives} of the form $\bigvee _{\mathsf{i}=0}^{\mathsf{n}}\phi $ or $\bigwedge _{\mathsf{i}=0}^{\mathsf{n}}\phi $.
      Notice that $\mathsf{n}$ denotes a parameter, which must be interpreted as a natural number.
      If arbitrary expressions for indices and iterations are allowed in the schema,
      then the satisfiability problem is undecidable, but we have identified in \cite{tab09,nested,jair2011}
      some subclasses for which this problem is decidable.
      The simplest of these classes is called \emph{regular}:
      it is defined by restricting both the indices of the propositions,
      that must be of the form $k$ or $\mathsf{n}+k$ where $k\in \mathbb{Z}$ and $\mathsf{n}$ is a variable,
      and iterations, that must be non-nested and of the form $\bigwedge _{\mathsf{i}=k}^{\mathsf{n}+l}\phi $ where $\mathsf{n}$ is a variable and $k,l\in \mathbb{Z}$.
      Decision procedures are designed in \cite{tab09,nested} and an implementation is available \cite{regstab}.\par     LTL and propositional schemata share many common features and
      trying to compare them precisely is a rather natural and, hopefully, fruitful idea.
      In both logics, interpretations can be viewed as arrays of propositional functions and the formulae relate the
      values of these functions at different states.
      The \emph{indices} of the propositions in the schematic case may be viewed as the \emph{time} in LTL.
      Thus comparing the expressing powers and complexities of those two logics,
      and, if possible, defining translations from one logic to the other is a natural and potentially rewarding issue.
      Notice that there already exist several results relating LTL to other formalisms like 
      monadic second order logic via Büchi automata \cite{ltl-buchi},
      monadic first order logic over natural numbers \cite{temporal-completeness} or star-free regular languages \cite{star-free}.
      However, there is a fundamental difference between these languages and the logic of schemata:
      they deal with infinite objects (infinite interpretations in the case of LTL or first order logic over natural numbers,
      infinite words in the case of star-free regular languages), whereas schemata deal with intrinsically finite (but unbounded) interpretations.
      This subtle but important difference introduces  difficulties in the definition of such translations.
      This topic bears some similarities with the approach of \cite{nivat} where problems on Bücchi automata are reduced to problems on finite automata
      by using the ultimately periodic property of ${\omega}$-regular languages.

  Note that finite interpretation is sometimes a desired feature:
      restricting LTL to finite traces has been considered in \cite{truncated-ltl},
      and has applications in, e.g., planning or runtime verification \cite{finite-ltl-original,ltl-finite-traces,finite-ltl-review}.
      It can be argued that the use of LTL in such contexts is a bit overkilling.
      Indeed, often, rather than considering finite traces per se, the
      preferred approach is to turn them into infinite traces by infinitely repeating the last state.
      It seems to us that it would be more natural to use schemata for such applications.
      In the present work, it is shown that doing so entails no loss in expressive power.\par     In the present paper, we show that LTL is equivalent to a particular subclass of regular schemata, referred to as \emph{sequential}.
      More precisely, we define functions effectively translating formulae from one
      logic into the other and show that this transformation preserves satisfiability.
      We believe that these results are interesting from a theoretical perspective since they
      provide useful information about the expressive power of the respective formalisms.
      Furthermore they allow to import the complexity results of LTL into schemata.
      From a practical point of view, the existence of a polynomial reduction from a class of
      propositional schemata into LTL allows one to benefit from the many existing efficient decision procedures for this logic
      (tableaux methods, e.g. \cite{wolper,one-pass},
      resolution-based methods, e.g. \cite{ltl-resolution},
      or reductions to model checking, e.g. \cite{ltl-reduction-experiments,ltl-antichains}),
      implementations \cite{logic-workbench,trp,nusmv,ltl-antichains}
      and experimentation tools \cite{lotrec}.
      Conversely, the reverse reduction might give further ideas for the design of new techniques to decide LTL satisfiability.
      In particular, since a \textsc{Dpll}-based procedure exists for regular schemata \cite{nested}, it might help to design such a procedure for LTL.
      On another hand, this reduction is very reminiscent of the translation from LTL
      to propositional logic encountered in bounded model checking (BMC) \cite{bmc}.
      Contrarily to BMC however, our reduction is \emph{complete}, it might thus give new ideas to achieve completeness in BMC.
      \par     The paper is structured as follows.
      In Section \ref{latex_lib_label_64} we define LTL and the logic of propositional schemata.
      In Section \ref{latex_lib_label_65} we show how to relate the interpretations of both formalisms.
      A polynomial algorithm transforming any sequential schema into an equivalent LTL formula is presented in Section \ref{latex_lib_label_27},
      and Section \ref{latex_lib_label_66} tackles the reverse translation, i.e. from LTL formulae to schemata.
      Section \ref{latex_lib_label_76} presents the results about first experiments with those translations
      and sketches the possible improvements inspired by those experiments.
      Section \ref{latex_lib_label_69} presents pros and cons of each logic, and make a very informal comparison of how LTL procedures
      behave on schema modulo the given translation, and, conversely, how schemata procedures behave on LTL formulae.
      Of course, with the given translations and the usual reduction of model checking to satisfiability, one can do model checking with schemata,
      Section \ref{latex_lib_label_75} gives an example of such model checking.
      Finally, Section \ref{latex_lib_label_67} briefly concludes our work.\par     
      \section{Definitions and notations}\label{latex_lib_label_64}In  the  following,  
      $\phi , \phi _{1}, \phi _{2}$  denote  LTL formulae,
      $s, s_{1}, s_{2}$  denote  schemata,
      ${\sigma}$  denotes  an  LTL or propositional interpretation,
      $\mathfrak{I}, \mathfrak{M}$ denote schema interpretations,
      $e, f, g$ denote (Presburger)  arithmetic expressions,
      $\mathsf{n}, \mathsf{i}$ denote arithmetic variables
      ($\mathsf{n}$ will be used for a free  arithmetic  variable (``parameter'') and $\mathsf{i}$ for a bound one).
      Remark  that $\mathsf{n}, \mathsf{i}$  are  written in sans serif 
      in order to distinguish them from meta variables denoting natural numbers,
      that will be written $n,i$.

  Both LTL and schemata have propositional logic as a common basis.
          Furthermore, in both languages, propositional variables are accompanied with a natural number
      (an instant in the case of LTL, an index for schemata).
      So instead of defining, as in classical propositional logic,
      an interpretation as a function mapping each propositional variable to a truth value,
      we rather define interpretations as functions mapping \emph{pairs} of propositional variables \emph{and natural numbers} to truth values.
      Formally:
  \begin{definition2}
    \label{latex_lib_label_50}Let $\mathcal{P}$ be a set of \emph{propositional variables}.
        A \emph{propositional interpretation} over $\mathcal{P}$ is a function from $\mathcal{P}\times \mathbb{N}$ to $\{ \mathrm{true}, \mathrm{false}\} $.
  \end{definition2}
  \begin{example2}
    Let $\mathcal{P}=\{ p, q\} $. Then ${\sigma}$ s.t. 
        $\sigma ( p, 0) =\mathrm{true}$,
        $\sigma ( q, 0) =\mathrm{false}$,
        $\sigma ( p, 1) =\mathrm{false}$,
        $\sigma ( q, 1) =\mathrm{false}$,
        $\sigma ( p, 2) =\mathrm{true}$,
        $\sigma ( q, 2) =\mathrm{true}$,
        and, for any $k>2$:
        $\sigma ( p, k) =\mathrm{true}$ and $\sigma ( q, k) =\mathrm{false}$,
        is a propositional interpretation.
  \end{example2}
   An interpretation ${\sigma}$ is represented by the set of all pairs (variable, natural number) that are true in ${\sigma}$.
          Most of the time we do not need to make that set explicit.
      For instance, when interpreting a given formula ${\phi}$, it will be implicitly assumed 
      that we consider only interpretations over sets that contain the variables of ${\phi}$.\subsection{LTL}
        We now recall the syntax and semantics of LTL.
  \begin{definition2}
          The syntax of LTL formulae over the set of propositional variables $\mathcal{P}$ is given by the following grammar:
          \begin{displaymath}
      \phi ::=\top \thinspace |\thinspace \mathcal{P}\thinspace |\thinspace \neg \phi \thinspace |\thinspace \phi \land \phi \thinspace |\thinspace \mathrm{X}\phi \thinspace |\thinspace \phi \mathrm{U}\phi 
    \end{displaymath}
  \end{definition2}

  $\mathrm{X}\phi $ means that ${\phi}$ holds at the next instant (``$\mathrm{X}$'' for neXt).
        $\phi \mathrm{U}\psi $ means that ${\phi}$ holds until ${\psi}$ holds (``$\mathrm{U}$'' for Until).
        We will also use the following abbreviations: $\mathrm{F}\phi \stackrel{{\mbox{\tiny{def}}}}{=}\top \mathrm{U}\phi $ and $\mathrm{G}\phi \stackrel{{\mbox{\tiny{def}}}}{=}\neg \mathrm{F}\neg \phi $,
        meaning respectively ``${\phi}$ eventually holds'' and ``${\phi}$ always holds''.
        The abbreviations $\lor $, $\Rightarrow $ and $\Leftrightarrow $ are defined as usual
        (the naive elimination of $\Leftrightarrow $ is exponential but it can be made linear by using 
        renaming of subformulae as usual, which preserves satisfiability).
        See \cite{ltl} for details.

  LTL formulae are usually interpreted over infinite paths in a transition system,
        together with a labelling that maps every state to a set of propositional variables.
        Such sequences are often called \emph{computations} or \emph{behaviours}.
        We will simply call them \emph{LTL interpretations}.
        For uniformity, we define formally an LTL interpretation as a propositional interpretation in the sense 
        of Definition \ref{latex_lib_label_50}
        (we do not make explicit the notions of states, transition systems and labelling).
        \begin{example2}
    The interpretation $\{ p, q\} \rightarrow \{ p\} \rightarrow \{ q\} \rightarrow \{ p, q\} \rightarrow \{ \} \rightarrow \{ \} \rightarrow \mbox{\dots }$
          is formally represented as the function ${\sigma}$ s.t.
          \begin{displaymath}
      \begin{gathered}
        \sigma ( p, 0) =\mathrm{true}\\
        \sigma ( q, 0) =\mathrm{true}\\
        \sigma ( p, 1) =\mathrm{true}\\
        \sigma ( q, 1) =\mathrm{false}\\
        \sigma ( p, 2) =\mathrm{false}\\
        \sigma ( q, 2) =\mathrm{true}\\
        \sigma ( p, 3) =\mathrm{true}\\
        \sigma ( q, 3) =\mathrm{true}\\
        \sigma ( p, 4) =\mathrm{false}\\
        \sigma ( q, 4) =\mathrm{false}\\
      \end{gathered}
    \end{displaymath}
  \end{example2}
        Then $\sigma ( t) $ denotes the set of variables $p$ that are true at time $t$, i.e. such that $( p, t) \in \sigma $ (in the previous example, $\sigma ( 0) =\{ p, q\} $, $\sigma ( 1) =\{ p\} $, etc.).
        The satisfaction relation of an LTL formula ${\phi}$ under such an interpretation ${\sigma}$ is defined w.r.t. an instant $t$, written $\sigma , t\models \phi $.
        This means that the formula ${\phi}$ holds at time $t$.
  \begin{definition2}
    \label{latex_lib_label_52}Let ${\phi}$ be an LTL formula, ${\sigma}$ be a propositional interpretation  and $t\in \mathbb{N}$.
          The relation $\sigma , t\models \phi $ is inductively defined as follows:
          \begin{displaymath}
      \begin{aligned}
        \sigma , t & \models \top \\
        \sigma , t & \models p\ \mbox{iff}\ ( p, t) \in \sigma \\
        \sigma , t & \models \neg \phi \ \mbox{iff}\ \sigma , t\not \models \phi \\
        \sigma , t & \models \phi _{1}\land \phi _{2}\ \mbox{iff}\ \sigma , t\models \phi _{1}\ \mbox{and}\ \sigma , t\models \phi _{2}\\
        \sigma , t & \models \mathrm{X}\phi \ \mbox{iff}\ \sigma , t+1\models \phi \\
        \sigma , t & \models \phi _{1}\mathrm{U}\phi _{2}\ \mbox{iff}\ \exists k\in \mathbb{N}\mbox{ s.t. }\forall i\in \mathbb{N}, i<k\Rightarrow \sigma , t+i\models \phi _{1}\mbox{ and }\sigma , t+k\models \phi _{2}\\
      \end{aligned}
    \end{displaymath}
          The notation $\sigma \models \phi $ means that ${\phi}$ is true in ${\sigma}$ at time $0$.
          \end{definition2}

  A fundamental property of LTL is the ``ultimately periodic model property''.
        Namely, if an LTL formula is satisfiable, then it is satisfiable on some ultimately periodic interpretation.
  \begin{definition2}
    \label{latex_lib_label_2}An \emph{ultimately periodic (``UP'') interpretation} is an LTL interpretation ${\sigma}$ s.t.
          there exist $k,l\in \mathbb{N}$ s.t. $l>0$ and for all $m\geq k$, $\sigma ( m) =\sigma ( m+l) $.
          The sequence $\sigma ( 0) \ldots \sigma ( k-1) $ is the \emph{prefix} of ${\sigma}$ and $\sigma ( k) \ldots \sigma ( k+l-1) $ its \emph{loop},
          $k$ is the \emph{prefix index} and $l$ is the \emph{period}.
  \end{definition2}
  \begin{theorem}[\cite{ltl-complexity}]
    \label{latex_lib_label_1}Any satisfiable LTL formula has a UP model.
  \end{theorem}
  This important result allows to focus exclusively on \emph{finite} sets of instants.
              Indeed it is sufficient to give the values of a UP interpretation for time $0$ to $k+l$.
        Other values until ${\omega}$ can then be computed.
  \begin{example2}
    Figure \ref{latex_lib_label_70} represents a UP model of $\mathrm{G}\mathrm{F}p$.
  \end{example2}
  \begin{figure}[h]
    \begin{displaymath}
      \mbox{\input{exupm.tex}}
    \end{displaymath}
    \caption{A UP model of $\mathrm{G}\mathrm{F}p$}\label{latex_lib_label_70}
  \end{figure}
  \subsection{Schemata}We now recall the syntax and semantics of schemata 
          (for simplicity, the considered definitions are slightly more restrictive than the ones of \cite{tab09}).
          Let $\mathcal{\MakeUppercase{\ensuremath{e}}}$ be the set of Presburger arithmetic expressions,
          i.e. terms built over a countably infinite set of \emph{arithmetic variables} $\mathcal{X}$
          and on the signature containing $0$, $\mathrm{succ}$, $+$ and possibly all the constant symbols in $\mathbb{N}$\footnote{Such constants may be encoded in unary,
          as terms of the form $\mathrm{succ}^{k}( 0) $ but also in binary, as sequences of digits.
          As we shall see, the choice between the two encodings has a significant influence on the complexity of the translation: polynomial translation of schemata into LTL is feasible only if numbers are encoded in unary.}.
          As usual a term is \emph{ground} iff it contains no variable.
          Notice that every ground expression will be considered the same as the natural number it represents.
  \begin{definition2}
    \label{latex_lib_label_72}The syntax of schemata over the set of propositional variables $\mathcal{P}$ is given by the following grammar:
            \begin{displaymath}
      s\thinspace ::=\thinspace \top \thinspace |\thinspace p_{e}\thinspace |\thinspace \neg s\thinspace |\thinspace s\land s\thinspace |\thinspace {\displaystyle \bigwedge _{\mathsf{i}=0}^{\mathsf{n}-1}s}
    \end{displaymath}
  \end{definition2}
   where $p\in \mathcal{P}$, $e\in \mathcal{\MakeUppercase{\ensuremath{e}}}$ and $\mathsf{i}, \mathsf{n}\in \mathcal{X}$.
          $\bigvee _{\mathsf{i}=0}^{\mathsf{n}-1}s$ is defined as ${\neg}$$\bigwedge _{\mathsf{i}=0}^{\mathsf{n}-1}\neg s$ and $\lor $, $\Rightarrow $ and $\Leftrightarrow $ are defined as usual.
  \begin{example2}
    $p_{0}\land \bigwedge _{\mathsf{i}=0}^{\mathsf{n}-1}( p_{\mathsf{i}}\Rightarrow p_{\mathsf{i}+1}) \land \neg p_{\mathsf{n}}$ 
            and $\bigwedge _{\mathsf{i}=0}^{\mathsf{n}-1}p_{\mathsf{i}}\land \bigvee _{\mathsf{i}=0}^{\mathsf{n}-1}\neg p_{\mathsf{i}}$ are schemata.
  \end{example2}
  \begin{remark2}
    This definition is less general than the one originally introduced in \cite{tab09} because all integers occurring in the schema must be positive
            (we consider Presburger arithmetic instead of linear arithmetic).
            This was not the case in \cite{tab09}, but it is easy to check that both formalisms have exactly the same expressive power.
            Furthermore the iterations are here restricted to go from $0$ to $\mathsf{n}-1$.
            Once again this is not restrictive w.r.t. to the expressive power,
            but it allows to get rid of tedious additional restrictions that would be needed otherwise.
            \end{remark2}
           Schemata of the form $p_{e}$ are called \emph{indexed propositions},
          and those of the form $\bigwedge _{\mathsf{i}=0}^{\mathsf{n}-1}s$ 
                  are called \emph{iterated  conjunctions} 
           or simply \emph{iterations}.
                  The variable  $\mathsf{i}$  is  \emph{bound} in $\bigwedge _{\mathsf{i}=0}^{\mathsf{n}-1}s$.
                  The essential point of schemata is that iterations are \emph{symbolic expressions}:
          $\mathsf{n}$ is a formal variable, called a \emph{parameter}, not a meta variable denoting any number.
          From now on, we assume that all schemata have only one parameter called $\mathsf{n}$.
          This is not restrictive for the scope of this paper (see \cite{schema-resolution}).\par         A schema is interpreted by first  giving a  value to the parameter -- which gives raise to a propositional formula  ${\phi}$,
          called an ``instance'' of the schema -- and then by giving a value to the propositional variables of ${\phi}$.
          Note that a schema has an infinite set of instances.
                  If $s$ is a schema or an arithmetic expression,
          $\mathsf{i}$ is an arithmetic variable and $e$ is an arithmetic expression,
          then $s$$[e/\mathsf{i}]$ denotes the expression obtained from $s$ by replacing every free occurrence of $\mathsf{i}$ by $e$.
          Note that, if $e$ is ground and $s$ is an arithmetic expression containing only the variable $\mathsf{i}$, then $s$$[e/\mathsf{i}]$ is a ground arithmetic expression,
          i.e. a natural number.
          Then:
  \begin{definition2}
    \label{latex_lib_label_23}Let $s$ be a schema of parameter $\mathsf{n}$ and $m\in \mathbb{N}$.
            The \emph{instance} of $s$ w.r.t. $m$ is the propositional formula
            $\langle s\rangle _{m}$ inductively defined as follows:
            \begin{displaymath}
      \begin{aligned}
        \langle p_{e}\rangle _{m} & \stackrel{{\mbox{\tiny{def}}}}{=}p_{e[m/\mathsf{n}]}\\
        \langle \neg s\rangle _{m} & \stackrel{{\mbox{\tiny{def}}}}{=}\neg \langle s\rangle _{m}\\
        \langle s_{1}\land s_{2}\rangle _{m} & \stackrel{{\mbox{\tiny{def}}}}{=}\langle s_{1}\rangle _{m}\land \langle s_{2}\rangle _{m}\\
        \langle {\displaystyle \bigwedge _{\mathsf{i}=0}^{\mathsf{n}-1}s}\rangle _{m} & \stackrel{{\mbox{\tiny{def}}}}{=}\top \mbox{ if $m=0$}\\
        \langle {\displaystyle \bigwedge _{\mathsf{i}=0}^{\mathsf{n}-1}s}\rangle _{m} & \stackrel{{\mbox{\tiny{def}}}}{=}\langle s[0/\mathsf{i}]\rangle _{m}\land \mbox{\dots }\land \langle s[m-1/\mathsf{i}]\rangle _{m}\mbox{ otherwise}\\
      \end{aligned}
    \end{displaymath}
            \end{definition2}
  \begin{example2}
    \label{latex_lib_label_71}
    \begin{displaymath}
      \begin{aligned}
        \langle p_{0}\land \bigwedge _{\mathsf{i}=0}^{\mathsf{n}-1}( p_{\mathsf{i}}\Rightarrow p_{\mathsf{i}+1}) \land \neg p_{\mathsf{n}}\rangle _{0} & =p_{0}\land \neg p_{0}\\
      \end{aligned}
    \end{displaymath}
    \begin{displaymath}
      \begin{aligned}
        \langle p_{0}\land \bigwedge _{\mathsf{i}=0}^{\mathsf{n}-1}( p_{\mathsf{i}}\Rightarrow p_{\mathsf{i}+1}) \land \neg p_{\mathsf{n}}\rangle _{1} & =p_{0}\land ( p_{0}\Rightarrow p_{1}) \land \neg p_{1}\\
        \langle p_{0}\land \bigwedge _{\mathsf{i}=0}^{\mathsf{n}-1}( p_{\mathsf{i}}\Rightarrow p_{\mathsf{i}+1}) \land \neg p_{\mathsf{n}}\rangle _{2} & =p_{0}\land ( p_{0}\Rightarrow p_{1}) \land ( p_{1}\Rightarrow p_{2}) \land \neg p_{2}\\
      \end{aligned}
    \end{displaymath}
     etc.
  \end{example2}
  An instance is a usual propositional formula except that each variable is indexed with a natural number.
          So we just need a propositional interpretation to interpret this formula as usual:
  \begin{definition2}
    Let ${\phi}$ be a propositional formula whose variables are indexed by natural numbers,
            and ${\sigma}$ a propositional interpretation.
            Then $\sigma \models \phi $ is defined as usual by induction on the structure of ${\phi}$ with the exception that,
            for any indexed propositional variable $p_{k}$, $\sigma \models p_{k}$ iff $( p, k) \in \sigma $.
  \end{definition2}
  We thus define a \emph{schema interpretation} as a pair consisting of a propositional interpretation and a natural number.
                  \begin{definition2}
    A schema $s$ is \emph{true in a schema interpretation $( \sigma , n) $} iff $\sigma \models \langle s\rangle _{n}$.
            We also use the notation $\models $ for schemata:
              $\mathfrak{I}\models s$ iff the schema $s$ is true in the schema interpretation $\mathfrak{I}$.
  \end{definition2}
  \begin{example2}
    $p_{0}\land \bigwedge _{\mathsf{i}=0}^{\mathsf{n}-1}( p_{\mathsf{i}}\Rightarrow p_{\mathsf{i}+1}) \land \neg p_{\mathsf{n}}$ is unsatisfiable (see its set of instances in Example \ref{latex_lib_label_71})
            as well as $\bigwedge _{\mathsf{i}=0}^{\mathsf{n}-1}p_{\mathsf{i}}\land \bigvee _{\mathsf{i}=0}^{\mathsf{n}-1}\neg p_{\mathsf{i}}$;
            $p_{0}\land \bigwedge _{\mathsf{i}=0}^{\mathsf{n}-1}( p_{\mathsf{i}}\Rightarrow p_{\mathsf{i}+1}) $ is satisfiable.
  \end{example2}

  The satisfiability problem for schemata is undecidable in general \cite{tab09}.
          However various decidable classes are investigated in \cite{tab09,nested,jair2011}.
          In the following, we will focus on the translation of LTL from/to ``sequential'' schemata:
  \begin{definition2}
    \label{latex_lib_label_19}A schema is a \emph{sequential propositional schema} (``SPS'') iff all the following conditions hold:
            \begin{itemize}
      \item it contains no nested iteration (iterations in the scope of another iteration);
      \item every index of a variable outside an iteration is of the form $k$ or $\mathsf{n}+k$, where $k\in \mathbb{N}$ and $\mathsf{n}$ is the parameter;
      \item every index of a variable inside an iteration $\bigwedge _{\mathsf{i}=0}^{\mathsf{n}-1}s$ is of the form $\mathsf{i}+k$, where $k\in \mathbb{N}$.
    \end{itemize}
  \end{definition2}
  \begin{example2}
    $p_{0}\land \bigwedge _{\mathsf{i}=0}^{\mathsf{n}-1}( p_{\mathsf{i}}\Rightarrow p_{\mathsf{i}+1}) $, $p_{0}\land \bigwedge _{\mathsf{i}=0}^{\mathsf{n}-1}( p_{\mathsf{i}}\Rightarrow p_{\mathsf{i}+1}) \land \neg p_{\mathsf{n}}$
            and $\bigwedge _{\mathsf{i}=0}^{\mathsf{n}-1}p_{\mathsf{i}}\land \bigvee _{\mathsf{i}=0}^{\mathsf{n}-1}\neg p_{\mathsf{i}}$ are SPS;
            $\bigwedge _{\mathsf{i}=0}^{\mathsf{n}}p_{\mathsf{i}}\land \bigvee _{\mathsf{i}=0}^{\mathsf{n}}\neg p_{\mathsf{i}}$, $p_{2\mathsf{n}}\land \bigwedge _{\mathsf{i}=0}^{\mathsf{n}-1}p_{\mathsf{i}}$, $\bigwedge _{\mathsf{i}=0}^{\mathsf{n}-1}p_{2\mathsf{n}}$, and $\bigwedge _{\mathsf{i}=0}^{\mathsf{n}-1}p_{2\mathsf{i}}$
            are not.
  \end{example2}
  Informally, an SPS represents a structure which is sequentially repeated, $\mathsf{n}$ being considered as the length of the sequence.
          SPS belong to the class of ``regular'' schemata, for which the satisfiability problem is proved to be decidable in \cite{tab09}.
          \section{Translating interpretations}\label{latex_lib_label_65}In the next sections we will provide translations of LTL formulae into SPS and conversely.
      Some semantic translations underlie those syntactic ones.
      We make them explicit now in order to give preliminary insights.\subsection{From schemata to LTL}Consider a schema interpretation  $( \sigma , n) $.
        Given a schema interpretation $( \sigma , n) $, its first component ${\sigma}$ can already be considered as an LTL interpretation,
        but we still need to represent the second component $n$.
        This is done by using special LTL interpretations (which are also propositional interpretations) called ``initial segments'':
  \begin{definition2}
    Let ${\sigma}$ be a propositional interpretation over a set of variables $\mathcal{P}$.
          ${\sigma}$ is an \emph{initial segment} of length $k\in \mathbb{N}$ for some $p\in \mathcal{P}$
          iff $( p, t) \in \sigma $ when $t<k$, and $( p, t) \not \in \sigma $ otherwise.
  \end{definition2}
  \begin{example2}
    Figure \ref{latex_lib_label_14} provides a graphical representation of an initial segment of length $4$ for $p$.
  \end{example2}
  \begin{figure}[h]
    \input{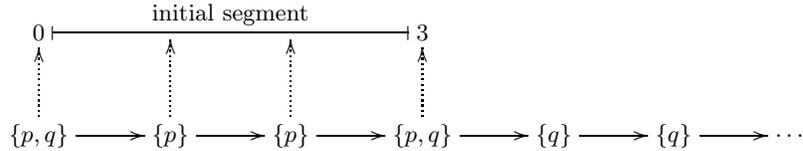}\caption{Initial segment of length $4$ for $p$.}\label{latex_lib_label_14}
  \end{figure}

        The key feature of initial segments is that they can be put in correspondence with natural numbers.
        Namely, we can associate a canonical initial segment to every natural number and a natural number to every initial segment.
              This correspondence allows us to define the following transformation for  schema interpretations:
  \begin{definition2}
    \label{latex_lib_label_53}Let $\mathcal{P}$ be a set of propositional variables and let $\mbox{``$\mathrm{t<\mathsf{n}}$''}\not \in \mathcal{P}$ be a propositional variable.
          Let $\mathfrak{I}=( \sigma , n) $ be a  schema interpretation over $\mathcal{P}$.
          Then $\llfloor \mathfrak{I}\rrfloor $ is the propositional interpretation (and thus also an LTL interpretation) over $\mathcal{P}\cup \{ \mathrm{t<\mathsf{n}}\} $
          which is an initial segment of length $n$ for $\mathrm{t<\mathsf{n}}$ and which is defined as ${\sigma}$ over $\mathcal{P}$.
          Conversely, $\llfloor .\rrfloor ^{-1}$ is the function that maps every initial segment ${\sigma}$ of length $n$ for $\mathrm{t<\mathsf{n}}$
          to the schema interpretation $( \tau , n) $ where ${\tau}$ is the restriction of ${\sigma}$ to $\mathcal{P}$.
          \end{definition2}
  \begin{example2}
    Let $\mathfrak{I}$ be the schema interpretation $( \{ p_{0}, q_{0}, p_{1}, p_{2}, q_{3}\} , 3) $.
          Then $\llfloor \mathfrak{I}\rrfloor =\{ p, q, \mathrm{t<\mathsf{n}}\} \rightarrow \{ p, \mathrm{t<\mathsf{n}}\} \rightarrow \{ p, \mathrm{t<\mathsf{n}}\} \rightarrow \{ q\} \rightarrow \{ \} \rightarrow \{ \} \rightarrow \mbox{\dots }$
          Conversely, let ${\sigma}$ be the LTL interpretation
          $\{ q, \mathrm{t<\mathsf{n}}\} \rightarrow \{ q, \mathrm{t<\mathsf{n}}\} \rightarrow \{ p, \mathrm{t<\mathsf{n}}\} \rightarrow \{ p, q, \mathrm{t<\mathsf{n}}\} \rightarrow \{ p\} \rightarrow \{ p\} \rightarrow \mbox{\dots }$,
          then $\llfloor \sigma \rrfloor ^{-1}=( \{ q_{0}, q_{1}, p_{2}, p_{3}, q_{3}, p_{4}, p_{5}, \mbox{\dots }\} , 4) $.
  \end{example2}

  The map $\llfloor .\rrfloor $ is a bijection between schema interpretations over $\mathcal{P}$
        and initial segments over $\mathcal{P}\cup \{ \mathrm{t<\mathsf{n}}\} $.
        Indeed, $\llfloor .\rrfloor ^{-1}$ is its inverse.
  \begin{remark2}
    An important difference between schemata and LTL is the fact that all
          interpretations of schemata are \emph{finite}, whereas those of LTL are \emph{infinite} (i.e. time is unbounded).
          Initial segments thus allow us to simulate finite models in LTL.
  \end{remark2}
   Finally notice that the set of initial segments can be specified in LTL as follows:
  \begin{proposition2}
    \label{latex_lib_label_31}Let ${\phi }_{<}^{\mathrm{t<\mathsf{n}}}$
          be the following formula: 
    \begin{displaymath}
      {\phi }_{<}^{\mathrm{t<\mathsf{n}}}\stackrel{{\mbox{\tiny{def}}}}{=}( \mathrm{t<\mathsf{n}}) \mathrm{U}\mathrm{G}( \neg \mathrm{t<\mathsf{n}}) 
    \end{displaymath}
          Then an LTL interpretation is a model of ${\phi }_{<}^{\mathrm{t<\mathsf{n}}}$ iff it is an initial segment for $\mathrm{t<\mathsf{n}}$.
  \end{proposition2}
  \begin{proof}
    An interpretation is a model of ${\phi }_{<}^{\mathrm{t<\mathsf{n}}}$ iff it makes $\mathrm{t<\mathsf{n}}$ true until ${\neg}$$\mathrm{t<\mathsf{n}}$ always holds.
          Let us write $k$ for the first instant where $\mathrm{t<\mathsf{n}}$ does not hold.
          Then this is equivalent to say that $\mathrm{t<\mathsf{n}}$ holds at time $t$ iff $t<k$.{\qed}
  \end{proof}
  We can also specify a proposition $\mathrm{\mbox{eq}^{n}}$ that is true only at time $n$.
        This is axiomatized by: 
        \begin{displaymath}
    \mathrm{Ax}_{\mathrm{t=\mathsf{n}}}\stackrel{{\mbox{\tiny{def}}}}{=}\mathrm{G}( \mathrm{t<\mathsf{n}}\land \neg \mathrm{X}( \mathrm{t<\mathsf{n}}) \Leftrightarrow \mathrm{X}( \mathrm{\mbox{eq}^{n}}) ) \land ( \neg \mathrm{t<\mathsf{n}}\Leftrightarrow \mathrm{\mbox{eq}^{n}}) 
  \end{displaymath}
        To improve readability, $\mathrm{\mbox{eq}^{n}}$ will be written $\mathrm{t=\mathsf{n}}$.
  \begin{proposition2}
    \label{latex_lib_label_62}Let ${\sigma}$ be an initial segment for $\mathrm{t<\mathsf{n}}$ of length $n$ s.t. $\sigma , 0\models \mathrm{Ax}_{\mathrm{t=\mathsf{n}}}$.
          Then  $\sigma , t\models \mathrm{t=\mathsf{n}}$ iff $t=n$.
  \end{proposition2}
  \begin{proof}
    By definition, $\mathrm{t<\mathsf{n}}$ holds at time $t$ iff $t<n$.
          If $n=0$ then $\mathrm{t<\mathsf{n}}$ never holds, in particular, $\mathrm{t<\mathsf{n}}$ does not hold at time $0$.
          Since $\sigma \models \mathrm{Ax}_{\mathrm{t=\mathsf{n}}}$, ${\sigma}$ satisfies its second conjunct,
          and as $\mathrm{t<\mathsf{n}}$ does not hold at time $0$, $\mathrm{\mbox{eq}}^{n}$ (i.e. $t=n$) holds at time $0$.
          Furthermore, since ${\sigma}$ satisfies the first conjunct and $\mathrm{t<\mathsf{n}}$ never holds again, $t=n$ is never satisfied again.
          Suppose now $n\ne 0$, then there is indeed at least one instant s.t. $\mathrm{t<\mathsf{n}}$ holds.
          Thus $\mathrm{t<\mathsf{n}}$ holds at time $n-1$ and not at time $n$, which corresponds precisely to the first conjunct of $\mathrm{Ax}_{\mathrm{t=\mathsf{n}}}$.
          Furthermore $n$ is the only instant with this property hence the result.{\qed}
  \end{proof}
  \subsection{From LTL to  schemata}\label{latex_lib_label_57}The inverse translation is harder: embedding LTL into schemata means that we must represent the
        infinite interpretations of LTL using only schema interpretations, which are finite.
        Of course this is impossible in general. 
        However, as we are concerned with satisfiability, 
        we can make use of Theorem \ref{latex_lib_label_1} and restrict ourselves to UP interpretations.
        Since such interpretations can be finitely represented, we will be able to embed them into schema interpretations.
        The representation of UP interpretations within schemata is achieved via 
        particular schema interpretations called ``$2$-initial segments'':
  \begin{definition2}
    A  schema interpretation $\mathfrak{I}=( \sigma , n) $ is a \emph{$2$-initial segment} for a propositional variable $p$
          iff there exists $k\leq n$ s.t., for every $l\in \{ 0, \mbox{\dots }, n\} $, we have $( p, l) \in \sigma \Leftrightarrow l<k$.
          We call $k$ the \emph{short length} of $\mathfrak{I}$ and $n+1$ is its \emph{long length}.
  \end{definition2}
  \begin{example2}
    The schema interpretation $( \{ p_{0}, p_{1}, p_{2}\} , 5) $ is a $2$-initial segment w.r.t. $p$(see Figure \ref{latex_lib_label_28}).
          Its short length is $3$, its long length is $6$.
  \end{example2}
  \begin{figure}[h]
    \input{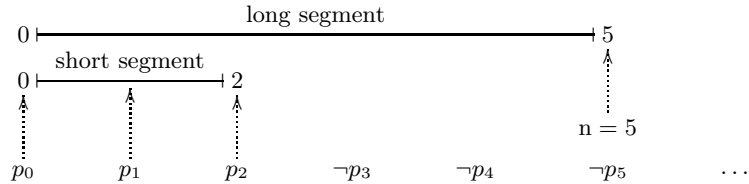}\caption{$2$-initial segment of short length $3$ and long length $6$ for $p$.}\label{latex_lib_label_28}
  \end{figure}
   We call this a $2$-initial segment because two initial segments are characterized:
        $\{ 0, \mbox{\dots }, k-1\} $ (characterized by $\mathfrak{I}$) and $\{ 0, \mbox{\dots }, n\} $ (characterized by $n$).
        Notice, however, that the segment $\{ 0, \mbox{\dots }, k-1\} $ is characterized by $p$ \emph{only below $n$},
        i.e. the value of $p$ is not specified above $n$.
        This is not a problem since we will not need such values in the translations.

  The notion of $2$-initial segment is useful because,
        much in the same way in which initial segments correspond to natural numbers,
        $2$-initial segments correspond to \emph{pairs} of different natural numbers.
        We can now define the following transformation for UP interpretations:
  \begin{definition2}
    \label{latex_lib_label_54}Let ${\sigma}$ be a UP interpretation of prefix index $k$ (i.e. the loop starts at time $k$)
          and of period $l$ over a set $\mathcal{P}$, and let $\mbox{``$\mathrm{pfx}$''}\not \in \mathcal{P}$ be a propositional variable.
          Then $\llceil \sigma \rrceil $ is the  schema interpretation $( \tau , k+l-1) $ 
          where ${\tau}$ is defined as an initial segment of length $k$ for $\mathrm{pfx}$ and preserving the value of ${\sigma}$ on $\mathcal{P}$.
  \end{definition2}
  \begin{example2}
    Let ${\sigma}$ be the UP interpretation of prefix index $2$ and period $3$ (totally) defined by:
          $\{ p, q, r\} \rightarrow \{ p\} \rightarrow \{ q, r\} \rightarrow \{ p, q\} \rightarrow \{ q, r\} $
          Then $\llceil \sigma \rrceil =( \{ \mathrm{pfx}_{0}, p_{0}, q_{0}, r_{0}, \mathrm{pfx}_{1}, p_{1}, q_{2}, r_{2}, p_{3}, q_{3}, q_{4}, r_{4}\} , 4) $.
          \end{example2}
  \begin{remark2}
    The map $\llceil .\rrceil $ embeds the prefix index and the period inside schema interpretations,
          but it is impossible to specify the fact that an interpretation is a UP interpretation:
          indeed this would require to express that the interpretation loops \emph{indefinitely}.
          Such a specification of an ``infinite'' behaviour cannot be achieved with schemata.
          This will not be a problem in the following because, when focusing on a given LTL formula,
          one only needs to specify this behaviour in the range $\{ 0, \mbox{\dots }, k+l-1\} $.\par       For similar reasons, $\llceil .\rrceil $ is \emph{not} a bijection in general, unlike $\llfloor .\rrfloor $.
          It is actually a bijection between UP interpretations and $2$-initial segments
          if we restrict the latter to the values assigned to variables whose index is between $0$ and $k+l-1$.
          This will indeed be the case in our reduction since, as just explained, we will not need the values for other indices.
          Then $\llceil .\rrceil ^{-1}$ is defined as follows:
          \end{remark2}
  \begin{definition2}
    \label{latex_lib_label_59}Let $( \sigma , n) $ be a $2$-initial segment for $\mathrm{pfx}$.
          Then $\llceil \sigma , n\rrceil ^{-1}$ is defined as the unique UP interpretation such that:
          \begin{itemize}
      \item its prefix is the set of instants s.t. $\mathrm{pfx}$ holds in $\mathfrak{I}$;
      \item its period $l$ is $n-k+1$, where $k$ is the prefix index;
      \item for all $p\ne \mathrm{pfx}$ and all $t\leq n$, $( p, t) \in \llceil \mathfrak{I}, n\rrceil ^{-1}$ iff $( p, t) \in \mathfrak{I}$.
    \end{itemize}
  \end{definition2}
  \begin{example2}
    Let $\mathfrak{I}=( \{ \mathrm{pfx}_{0}, p_{0}, \mathrm{pfx}_{1}, q_{1}, p_{2}, p_{3}, q_{3}\} , 3) $.
          Then $\llceil \mathfrak{I}\rrceil ^{-1}$ is the UP interpretation of prefix index $2$ and period $2$ defined by
          $\llceil \mathfrak{I}\rrceil ^{-1}=\{ p\} \rightarrow \{ q\} \rightarrow \{ p\} \rightarrow \{ p, q\} \rightarrow \mbox{\dots }$
          where the contents of the dots can be retrieved by the UP property of the interpretation.
  \end{example2}

  Finally, $2$-initial segments can be specified using schemata:
  \begin{proposition2}
    \label{latex_lib_label_35}Let ${\mathrm{s}}_{\leq }^{\mathrm{pfx}}$ be the
          following SPS:
            \begin{displaymath}
      {\mathrm{s}}_{\leq }^{\mathrm{pfx}}\stackrel{{\mbox{\tiny{def}}}}{=}\neg \mathrm{pfx}_{\mathsf{n}}\land \bigwedge _{\mathsf{i}=0}^{\mathsf{n}-1}( \mathrm{pfx}_{\mathsf{i}+1}\Rightarrow \mathrm{pfx}_{\mathsf{i}}) 
    \end{displaymath}
            Then a
          schema interpretation is a model of ${\mathrm{s}}_{\leq }^{\mathrm{pfx}}$ iff it is a $2$-initial segment for $\mathrm{pfx}$.
  \end{proposition2}
  \begin{proof}
    Let $\mathfrak{I}=( \sigma , n) $ be a model of ${\mathrm{s}}_{\leq }^{\mathrm{pfx}}$.
          For any $k\in \mathbb{N}$ s.t. $\mathrm{pfx}_{k}$ holds, $\mathrm{pfx}_{k'}$ holds for every $k'<k$, because ${\sigma}$ satisfies the second conjunct of ${\mathrm{s}}_{\leq }^{\mathrm{pfx}}$.
          Furthermore there is a maximal such $k<n$ , because $\mathrm{pfx}_{k}$ cannot hold at time $n$, by the first conjunct.
          Hence $\mathfrak{I}$ is indeed a $2$-initial segment.\par       Conversely, let $\mathfrak{I}=( \sigma , n) $ be a $2$-initial segment for $\mathrm{pfx}$ of short length $k$.
          Then, for every $l\in \{ 0, \mbox{\dots }, n\} $, $\mathrm{pfx}_{l}$ holds iff $l<k$.
          Since $k\leq n$, $\mathrm{pfx}_{n}$ cannot hold, hence the first conjunct is indeed satisfied.
          Furthermore for every $l\in \{ 0, \mbox{\dots }, n\} $, if $\mathrm{pfx}_{l+1}$ holds then $\mathrm{pfx}_{l}$ holds, hence the second conjunct is satisfied.{\qed}
  \end{proof}

  The beginning of the loop can be referred to by using a propositional variable ${\mathrm{eq}}_{\mathsf{i}}^{k}$,
        intended to be true only when $\mathsf{i}$ is equal to the prefix index $k$ of the interpretation.
        This can be axiomatized as follows:
        \begin{displaymath}
    \mathrm{Ax}_{\mathsf{i}=k}\stackrel{{\mbox{\tiny{def}}}}{=}( \neg \mathrm{pfx}_{0}\Leftrightarrow {\mathrm{eq}}_{0}^{k}) \land \bigwedge _{\mathsf{i}=0}^{\mathsf{n}-1}( \mathrm{pfx}_{\mathsf{i}}\land \neg \mathrm{pfx}_{\mathsf{i}+1}\Leftrightarrow {\mathrm{eq}}_{\mathsf{i}+1}^{k}) 
  \end{displaymath}
        To improve readability, ${\mathrm{eq}}_{\mathsf{i}}^{k}$ will be written ``$\mathsf{i}=k$''.
  \begin{proposition2}
    \label{latex_lib_label_73}Let $\mathfrak{I}$ be a $2$-initial segment of short length $k$ for $\mathrm{pfx}$ s.t. $\mathfrak{I}\models \mathrm{Ax}_{\mathsf{i}=k}$.
          Then, for every $i\in \{ 0, \mbox{\dots }, n\} $, $\mathfrak{I}\models {\mathrm{eq}}_{i}^{k}$ iff $i=k$.
  \end{proposition2}
  \begin{proof}
    If $k=0$, then $\mathfrak{I}\not \models \mathrm{pfx}_{0}$ hence $\mathsf{i}\models {\mathrm{eq}}_{0}^{k}$, by the first conjunct of $\mathrm{Ax}_{\mathsf{i}=k}$.
          Furthermore $\mathfrak{I}\not \models {\mathrm{eq}}_{i+1}^{k}$ for any $i\in \{ 0, \mbox{\dots }, n-1\} $ because $\mathrm{pfx}_{i}$ does not hold and by the second conjunct of $\mathrm{Ax}_{\mathsf{i}=k}$.\par       If $k>0$, then $\mathfrak{I}\models \mathrm{pfx}_{k-1}$ and $\mathfrak{I}\not \models \mathrm{pfx}_{k}$ hence, by the second conjunct, $\mathfrak{I}\models {\mathrm{eq}}_{k}^{k}$.
          Furthermore, no other instant $l$ between $0$ and $n$ has the property that $\mathfrak{I}\models \mathrm{pfx}_{l-1}$ and $\mathfrak{I}\not \models \mathrm{pfx}_{l}$, hence the equivalence.{\qed}
  \end{proof}
  \section{Embedding SPS in LTL}\label{latex_lib_label_27}We now show how SPS can be translated into LTL:
    given an SPS $s$, we build an LTL formula $\left\lfloor s\right\rfloor $ which is satisfiable iff $s$ is satisfiable.
    Then we show that the size of $\left\lfloor s\right\rfloor $ is polynomial or exponential w.r.t. the size of $s$, depending on the encoding of natural numbers 
    (in the arithmetic expressions occurring in $s$).
    As LTL satisfiability is in PSPACE, we can thus conclude that the satisfiability of SPS
    is also in PSPACE when numbers are encoded in unary.\subsection{The $\left\lfloor .\right\rfloor $ transformation}\label{latex_lib_label_18} The main desideratum of $\left\lfloor .\right\rfloor $ is that for every model $\mathfrak{M}$ of an SPS $s$,
    the interpretation $\llfloor \mathfrak{M}\rrfloor $ (Definition \ref{latex_lib_label_53}) is a model of $\left\lfloor s\right\rfloor $.
    An example is shown on Figure \ref{latex_lib_label_29} 
    (we represent LTL interpretations as sequences of sets of propositional variables, 
    instead of sets of pairs (variable, number), as they are formally defined;
    similarly, schema interpretations are represented as the set of true indexed propositions).
  \begin{figure}[h]
    \input{exspec.tex}\caption{Specification of $\left\lfloor .\right\rfloor $: example.}\label{latex_lib_label_29}
  \end{figure}

  By Proposition \ref{latex_lib_label_31},
    every interpretation such that ${\phi }_{<}^{\mathrm{t<\mathsf{n}}}$ (defined in Proposition \ref{latex_lib_label_31})
    holds is an initial segment of length $n$ for a propositional variable ``$\mathrm{t<\mathsf{n}}$''.
    Furthermore, $\mathrm{Ax}_{\mathrm{t=\mathsf{n}}}$ (defined after Proposition \ref{latex_lib_label_31}) enables to use the variable ``$\mathrm{t=\mathsf{n}}$''.
    Our translation thus includes those formulae.
  \begin{definition2}
    \label{latex_lib_label_8}Let $s$ be an  SPS.
      Then $\left\lfloor s\right\rfloor $ is an LTL formula defined as 
      $\left\lfloor s\right\rfloor \stackrel{{\mbox{\tiny{def}}}}{=}\left\lfloor s\right\rfloor _{\mathrm{prop}}\land {\phi }_{<}^{\mathrm{t<\mathsf{n}}}\land \mathrm{Ax}_{\mathrm{t=\mathsf{n}}}$
      where $\left\lfloor s\right\rfloor _{\mathrm{prop}}$ is inductively defined as follows:
      \begin{displaymath}
      \begin{aligned}
        \left\lfloor \top \right\rfloor _{\mathrm{prop}} & \stackrel{{\mbox{\tiny{def}}}}{=}\top \\
        \left\lfloor p_{k}\right\rfloor _{\mathrm{prop}} & \stackrel{{\mbox{\tiny{def}}}}{=}\mathrm{X}^{k}p\\
        \left\lfloor p_{\mathsf{n}+k}\right\rfloor _{\mathrm{prop}} & \stackrel{{\mbox{\tiny{def}}}}{=}\mathrm{G}( \mathrm{t=\mathsf{n}}\Rightarrow \mathrm{X}^{k}p) \\
        \left\lfloor p_{\mathsf{i}+k}\right\rfloor _{\mathrm{prop}} & \stackrel{{\mbox{\tiny{def}}}}{=}\mathrm{X}^{k}p\\
        \left\lfloor \neg s\right\rfloor _{\mathrm{prop}} & \stackrel{{\mbox{\tiny{def}}}}{=}\neg \left\lfloor s\right\rfloor _{\mathrm{prop}}\\
      \end{aligned}
    \end{displaymath}
    \begin{displaymath}
      \begin{aligned}
        \left\lfloor s_{1}\land s_{2}\right\rfloor _{\mathrm{prop}} & \stackrel{{\mbox{\tiny{def}}}}{=}\left\lfloor s_{1}\right\rfloor _{\mathrm{prop}}\land \left\lfloor s_{2}\right\rfloor _{\mathrm{prop}}\\
        \left\lfloor {\displaystyle \bigwedge _{\mathsf{i}=0}^{\mathsf{n}-1}s}\right\rfloor _{\mathrm{prop}} & \stackrel{{\mbox{\tiny{def}}}}{=}\mathrm{G}( \mathrm{t<\mathsf{n}}\Rightarrow \left\lfloor s\right\rfloor _{\mathrm{prop}}) \\
      \end{aligned}
    \end{displaymath}
      where $k\in \mathbb{N}$, $\mathsf{i}\ne \mathsf{n}$, and $\mathrm{X}^{k}$${\phi}$
      $\stackrel{{\mbox{\tiny{def}}}}{=}\underbrace {\mathrm{X}\mbox{\dots }\mathrm{X}}_{\times k}\phi $.
      \end{definition2}
  \begin{example2}
    \label{latex_lib_label_68}We have: $\left\lfloor p_{0}\land \bigwedge _{\mathsf{i}=0}^{\mathsf{n}-1}( p_{\mathsf{i}}\Rightarrow p_{\mathsf{i}+1}) \land \neg p_{\mathsf{n}}\right\rfloor =p\land \mathrm{G}( \mathrm{t<\mathsf{n}}\Rightarrow p\Rightarrow \mathrm{X}p) \land \neg \mathrm{G}( \mathrm{t=\mathsf{n}}\Rightarrow p) \land {\phi }_{<}^{\mathrm{t<\mathsf{n}}}\land \mathrm{Ax}_{\mathrm{t=\mathsf{n}}}$. Notice that it would be equivalent to have
      $p\land \mathrm{G}( \mathrm{t<\mathsf{n}}\Rightarrow p\Rightarrow \mathrm{X}p) \land \mathrm{G}( \mathrm{t=\mathsf{n}}\Rightarrow \neg p) \land {\phi }_{<}^{\mathrm{t<\mathsf{n}}}\land \mathrm{Ax}_{\mathrm{t=\mathsf{n}}}$
      because $\mathrm{t=\mathsf{n}}$ holds at only one moment.
      This variation is interesting because it does not introduce any eventuality, and is thus easier to handle for LTL decision procedures.
      It can be generalized, e.g., by putting every schema into n.n.f. \emph{before} the translation,
      and then by defining a dedicated case for negative literals.
  \end{example2}
   Figure \ref{latex_lib_label_29} can now be updated into Figure \ref{latex_lib_label_17}.
  \begin{figure}[h]
    \input{exbigpicture.tex}\caption{Big picture for $\bigwedge _{\mathsf{i}=0}^{\mathsf{n}}( p_{\mathsf{i}}\lor q_{\mathsf{i}}) $ and one of its models.}\label{latex_lib_label_17}
  \end{figure}
  \subsection{Soundness and completeness of $\left\lfloor .\right\rfloor $.}
  \begin{theorem}
    \label{latex_lib_label_12}Let  $s$  be  a  SPS.
        Then $\llfloor .\rrfloor $  is a bijection between the models of $s$ and the  models of $\left\lfloor s\right\rfloor $.
        The inverse bijection is $\llfloor .\rrfloor ^{-1}$ (Definition \ref{latex_lib_label_53}).
  \end{theorem}

  This result is more interesting than just ``$s$ is satisfiable iff $\left\lfloor s\right\rfloor $ is satisfiable''. Indeed, not only does it provide
       more insights about the translation,
       but it also makes explicit the inverse transformation for interpretations,
       which is useful for model building.
  \begin{proof}
    Notice that $\llfloor .\rrfloor ^{-1}$ is well defined because every model of $\left\lfloor s\right\rfloor $ is an initial segment by Proposition \ref{latex_lib_label_31}.
          We still have to prove the following:
          \begin{enumerate}
      \item \label{latex_lib_label_79}for every model $\mathfrak{M}$ of $s$, $\llfloor \mathfrak{M}\rrfloor $ is a model of $\left\lfloor s\right\rfloor $;
      \item \label{latex_lib_label_80}for every model ${\sigma}$ of $\left\lfloor s\right\rfloor $, $\llfloor \sigma \rrfloor ^{-1}$ is a model of $s$.
    \end{enumerate}
          In the following, $\mathfrak{M}$ is a model of $s$, $n$ is the value given to $\mathsf{n}$ by $\mathfrak{M}$,
          ${\sigma}$ is a model of $\left\lfloor s\right\rfloor $ and $l$ is the length of $\left\lfloor s\right\rfloor $ (as an initial segment of $\mathrm{t<\mathsf{n}}$).
          Notice that, by definition, $\llfloor \mathfrak{M}\rrfloor $ coincides with the propositional part of $\mathfrak{M}$ on any propositional variable other that $\mathrm{t<\mathsf{n}}$.
          Similarly, ${\sigma}$ coincides with the propositional part of  $\llfloor \sigma \rrfloor ^{-1}$ on any propositional variable other that $\mathrm{t<\mathsf{n}}$.\par       We prove both properties simultaneously by induction on $\left\lfloor s\right\rfloor _{\mathrm{prop}}$:
          \begin{itemize}
      \item Suppose $s=p_{k}$, where $k\in \mathbb{N}$. 
              Then $\mathfrak{M}\models s$ implies $( p, k) \in \mathfrak{M}\subset \llfloor \mathfrak{M}\rrfloor $
              and by a straightforward induction on $k$ this implies that 
              $\llfloor \mathfrak{M}\rrfloor \models \mathrm{X}^{k}p$ which proves \ref{latex_lib_label_79}. 
              For \ref{latex_lib_label_80}, suppose that $\sigma \models \left\lfloor s\right\rfloor $, i.e.  $\sigma \models \mathrm{X}^{k}p$. 
              This easily entails that $( p, k) \in \sigma $, thus $( p, k) \in \llfloor \sigma \rrfloor ^{-1}$ and so $\llfloor \sigma \rrfloor ^{-1}\models p_{k}$.
      \item Suppose $s=p_{\mathsf{n}+k}$. 
              Then $\mathfrak{M}\models p_{\mathsf{n}+k}$ means that $( p, n+k) \in \mathfrak{M}\subset \llfloor \mathfrak{M}\rrfloor $.
              Consequently, $\mathrm{X}^{k}$$p$ is true at time $n$.
              Finally, by Proposition \ref{latex_lib_label_62}, $\mathrm{t=\mathsf{n}}$ is true only at time $n$ in $\llfloor \mathfrak{M}\rrfloor $.
              Thus, at any time when we have $\mathrm{t=\mathsf{n}}$, we have $\mathrm{X}^{k}$$p$;
              i.e. we have $\mathrm{t=\mathsf{n}}\Rightarrow \mathrm{X}^{k}p$ at any time;
              i.e. we have $\mathrm{G}( \mathrm{t=\mathsf{n}}\Rightarrow \mathrm{X}^{k}p) $.
              This proves \ref{latex_lib_label_79}.
              For \ref{latex_lib_label_80}, suppose $\sigma \models \left\lfloor s\right\rfloor $, i.e. $\sigma \models \mathrm{G}( \mathrm{t=\mathsf{n}}\Rightarrow \mathrm{X}^{k}p) $.
              Thus for every $t\in \mathbb{N}$, $\mathrm{t=\mathsf{n}}\Rightarrow \mathrm{X}^{k}p$ is true at time $t$ in ${\sigma}$.
              But we know that $\mathrm{t=\mathsf{n}}$ is true only at time $l$ (the length of $\left\lfloor s\right\rfloor $).
              Thus $p$ is true at time $l+k$.
              Hence $( p, l+k) \in \sigma $, thus $( p, l+k) \in \llfloor \sigma \rrfloor ^{-1}$, and since the value of $\mathsf{n}$ in $\llfloor \sigma \rrfloor ^{-1}$ is $l$, $\llfloor \sigma \rrfloor ^{-1}\models p_{\mathsf{n}+k}$.
      \item The case $s=p_{\mathsf{i}+k}$ is handled in the iteration cases (Lemma \ref{latex_lib_label_7}).
      \item Suppose $s=\neg s'$.
              For \ref{latex_lib_label_79}, if $\mathfrak{M}\models s$ then $\mathfrak{M}\not \models s'$. 
              But, as $\llfloor .\rrfloor ^{-1}$ is the inverse of $\llfloor .\rrfloor $, $\mathfrak{M}=\llfloor ( \llfloor \mathfrak{M}\rrfloor ) \rrfloor ^{-1}$.
              Thus $\llfloor ( \llfloor \mathfrak{M}\rrfloor ) \rrfloor ^{-1}\not \models s'$.
              By induction hypothesis, \ref{latex_lib_label_80} holds for $s'$, thus, by contraposition: $\llfloor \mathfrak{M}\rrfloor \not \models \left\lfloor s'\right\rfloor $.
              Consequently, $\llfloor \mathfrak{M}\rrfloor \models \neg \left\lfloor s'\right\rfloor $.
              For \ref{latex_lib_label_80}, suppose $\sigma \models \left\lfloor s\right\rfloor $, i.e. $\sigma \models \neg \left\lfloor s'\right\rfloor $. 
              Thus $\sigma \not \models \left\lfloor s'\right\rfloor $, i.e. $\llfloor \llfloor \sigma \rrfloor ^{-1}\rrfloor \not \models \left\lfloor s'\right\rfloor $.
              By induction hypothesis, \ref{latex_lib_label_79} holds for $s'$, so, by contraposition: $\llfloor \sigma \rrfloor ^{-1}\not \models s'$.
              Thus $\llfloor \sigma \rrfloor ^{-1}\models \neg s'$.
              \item Suppose $s=s_{1}\land s_{2}$.
              For \ref{latex_lib_label_79}, if $\mathfrak{M}\models s$ then $\mathfrak{M}\models s_{1}$ and $\mathfrak{M}\models s_{2}$ and one easily concludes by induction.
              For \ref{latex_lib_label_80}, if $\sigma \models \left\lfloor s\right\rfloor $, i.e. $\sigma \models \left\lfloor s_{1}\right\rfloor \land \left\lfloor s_{2}\right\rfloor $, then $\sigma \models \left\lfloor s_{1}\right\rfloor $ and $\sigma \models \left\lfloor s_{2}\right\rfloor $ and one can also conclude by induction.
      \item Suppose $s=\bigwedge _{\mathsf{i}=0}^{\mathsf{n}}s'$.
              We first prove the following intermediate lemma:\par         
      \begin{lemma2}
        \label{latex_lib_label_7}For every initial segment ${\sigma}$ of length $l$, and every $t\leq l$:
                $\sigma \models \left\lfloor s'[t/\mathsf{i}]\right\rfloor $ iff $\left\lfloor s'\right\rfloor $ holds in ${\sigma}$ at time $t$.
                \end{lemma2}
              \begin{proof}
        We prove both implications simultaneously by induction on the structure of $s'$:
                \begin{itemize}
          \item Suppose $s'=p_{\mathsf{i}+k}$;
                    thus $\left\lfloor s'\right\rfloor =\left\lfloor p_{\mathsf{i}+k}\right\rfloor =\mathrm{X}^{k}p$,
                    and $\left\lfloor s'[t/\mathsf{i}]\right\rfloor =\left\lfloor p_{t+k}\right\rfloor =\mathrm{X}^{t+k}p$.
                    For the first implication,
                    assume that $\sigma \models \left\lfloor s'[t/\mathsf{i}]\right\rfloor $, i.e. $\sigma \models \mathrm{X}^{t+k}p$, which is equivalent to $p\in \llfloor \mathfrak{M}\rrfloor ( t+k) $.
                    It is equivalent to say that $\mathrm{X}^{k}$$p$, i.e. $\left\lfloor s'\right\rfloor $, is true in ${\sigma}$ at time $t$.
                    This proves the first implication, and also the second as all reasoning steps are equivalences.
          \item As $s$ is strictly bound, there are no other base case (this is precisely why this restriction is essential).
          \item Suppose $s'=\neg s''$: then $\left\lfloor s'\right\rfloor =\neg \left\lfloor s''\right\rfloor $ and $\left\lfloor s'[t/\mathsf{i}]\right\rfloor =\neg \left\lfloor s''[t/\mathsf{i}]\right\rfloor $.\par           For the first implication, assume that $\sigma \models \left\lfloor s'[t/\mathsf{i}]\right\rfloor $, i.e. $\sigma \models \neg \left\lfloor s''[t/\mathsf{i}]\right\rfloor $.
                    Thus $\sigma \not \models \left\lfloor s''[t/\mathsf{i}]\right\rfloor $.
                    By the reverse implication of the induction hypothesis (more precisely by its contraposition),
                    this means that $\left\lfloor s''\right\rfloor $ does not hold in ${\sigma}$ at time $t$.
                    Consequently, ${\neg}$$\left\lfloor s''\right\rfloor $ holds in ${\sigma}$ at time $t$, hence the result.
                    Once again, the second implication is obtained by just reversing the reasoning.
          \item The proof for the  conjunction case is routine.
          \item As the schema is sequential, iterations cannot be nested, thus $s'$ cannot contain an iteration, hence there are no more cases.
                    \hfill${\Diamond}$
                    \end{itemize}
                \end{proof}
              Now we can get back to the iteration case of the main proof.
              For \ref{latex_lib_label_79}, if $\mathfrak{M}\models s$ then $\mathfrak{M}\models s'[t/\mathsf{i}]$ for every $t\in \mathbb{N}$ s.t. $0\leq t\leq n$ by definition of schemata semantics.
              And thus, by induction hypothesis, $\llfloor \mathfrak{M}\rrfloor \models \left\lfloor s'[t/\mathsf{i}]\right\rfloor $ for every such $t$.
              By Lemma \ref{latex_lib_label_7}, this means that $\left\lfloor s'\right\rfloor $ is true in $\llfloor \mathfrak{M}\rrfloor $ at any time $t$ s.t. $0\leq t\leq n$.
              From the semantics of LTL, it is obvious that $t\geq 0$
              so it is enough to say that $\left\lfloor s'\right\rfloor $ is true in $\llfloor \mathfrak{M}\rrfloor $ at any time $t\leq n$
              (notice that this would not be so simple if the schema was not simply iterated).
              This is equivalent to say that $\mathrm{t<\mathsf{n}}\Rightarrow \left\lfloor s'\right\rfloor $ is true at any time, hence the conclusion for \ref{latex_lib_label_79}.\par         For \ref{latex_lib_label_80}, suppose that $\sigma \models \left\lfloor s\right\rfloor $, i.e. $\sigma \models \mathrm{G}( \mathrm{t<\mathsf{n}}\Rightarrow \left\lfloor s\right\rfloor ) $.
              Then, by definition of LTL semantics, $\mathrm{t<\mathsf{n}}\Rightarrow \left\lfloor s\right\rfloor $ is true in ${\sigma}$ at any time.
              Furthermore $\mathrm{t<\mathsf{n}}$ is true only at time $l$ or below, hence $\left\lfloor s\right\rfloor $ is true in ${\sigma}$ at any time less or equal than $l$.
              We can then conclude using the reverse implication of Lemma \ref{latex_lib_label_7} and the semantics of schemata.
              {\qed}
              \end{itemize}
        \end{proof}
  \subsection{Consequences.}We then obviously have the expected result:
  \begin{corollary2}
    A SPS $s$ is satisfiable iff $\left\lfloor s\right\rfloor $ is satisfiable.
  \end{corollary2}
   Thus we indeed obtained an embedding of  SPS into LTL. Consequently we can use any LTL satisfiability solver to solve
      the satisfiability problem for SPS:
      we simply translate the input schema to LTL with $\left\lfloor .\right\rfloor $ and then launch the LTL solver on the output formula.
      Thus: 
  \begin{corollary2}
    The satisfiability problem for SPS can be reduced to the satisfiability problem for LTL.
  \end{corollary2}
      Notice furthermore that if the solver finds a model, then we can translate it back to a schema model
      using the inverse translation $\llfloor .\rrfloor ^{-1}$.\par     We can easily study the complexity of this transformation.
      For an object $x$ (schema, formula, arithmetic expression), let $\#x$ stand for the size of $x$ in number of symbols.
      Let $\#s$ denote the size of a schema $s$, in number of symbols,
      and let $\#_{\mathrm{int}}s$ denote the size of the biggest number occurring in $s$, expressed w.r.t. the size of $s$.
      This is to take into account the fact that numbers can be encoded either in unary or in binary:
      if they are encoded in binary then $\#_{\mathrm{int}}s=O( 2^{\#s}) $,
      but if they are encoded in unary then $\#_{\mathrm{int}}s=O( \#s) $.
      It may also happen that we consider only schemata whose biggest number is bounded by some constant; in such a case, we have $\#_{\mathrm{int}}s=O( 1) $.
      This case is worth considering since we may easily increase the size of a schema without increasing the numbers that occur in it.
      Then:
      \begin{proposition2}
    For every SPS $s$, we have $\#\left\lfloor s\right\rfloor =O( \#s.\#_{\mathrm{int}}s) $.
  \end{proposition2}
      \begin{proof}
    First, ${\phi }_{<}^{\mathrm{t<\mathsf{n}}}$ has a constant size.
        Then since the construction of ${\left\lfloor s\right\rfloor }_{\mathrm{prop}}$ is by induction on $s$, there are $\#s$ recursive calls.
        Each of those calls adds a number of symbols either constant or 
        proportional to some $k\in \mathbb{N}$ occurring in $s$ (all the cases with
        ``$\times k$''), i.e. at worst $\#_{\mathrm{int}}s$.{\qed}
  \end{proof}
      Consequently, $\left\lfloor .\right\rfloor $ is: 
      \begin{itemize}
    \item linear if numbers are bounded by constants; 
    \item quadratic if numbers are encoded in unary; 
    \item exponential if they are encoded in binary.
  \end{itemize}
      It is well-known that the satisfiability of LTL is in PSPACE \cite{ltl-complexity}, thus:
      \begin{theorem}
    \label{latex_lib_label_34}The satisfiability of SPS is in PSPACE  if numbers are encoded in unary or bounded
        by constants. It is in EXPSPACE if numbers are encoded in binary.
  \end{theorem}
       This result improves over the one of \cite{lata2010}, 
      where the satisfiability of \emph{regular} schemata is proved to be in EXPSPACE (resp. $2$-EXPSPACE),
      if numbers are encoded in unary (resp. binary).
      Of course Theorem \ref{latex_lib_label_34} only deals with sequential schemata,
      but both classes are close enough so that we conjecture that the satisfiability of regular schemata is also in PSPACE.\section{Embedding LTL in SPS}\label{latex_lib_label_66}We now tackle the reverse embedding, i.e. we translate LTL to SPS.\subsection{A first \emph{faulty} translation: finiteness vs infiniteness.}We provide a first, intuitive but faulty, translation:
        \begin{definition2}
    \label{latex_lib_label_25}Let ${\phi}$ be an LTL formula.
          Then $\lceil \phi \rceil $ is a schema defined as $\lceil \phi \rceil _{0}$, where $\lceil \phi \rceil _{e}$ is
          inductively defined for any expression $e$ as follows:
          \begin{displaymath}
      \begin{aligned}
        \lceil \top \rceil _{e} & \stackrel{{\mbox{\tiny{def}}}}{=}\top \\
        \lceil p\rceil _{e} & \stackrel{{\mbox{\tiny{def}}}}{=}p_{e}\\
        \lceil \neg \phi \rceil _{e} & \stackrel{{\mbox{\tiny{def}}}}{=}\neg \lceil \phi \rceil _{e}\\
        \lceil \phi _{1}\land \phi _{2}\rceil _{e} & \stackrel{{\mbox{\tiny{def}}}}{=}\lceil \phi _{1}\rceil _{e}\land \lceil \phi _{2}\rceil _{e}\\
        \lceil \mathrm{X}\phi \rceil _{e} & \stackrel{{\mbox{\tiny{def}}}}{=}\lceil \phi \rceil _{e+1}\\
        \lceil \phi _{1}\mathrm{U}\phi _{2}\rceil _{e} & \stackrel{{\mbox{\tiny{def}}}}{=}{\displaystyle \bigvee _{\mathsf{i}=e}^{\mathsf{n}}\left( \lceil \phi _{2}\rceil _{\mathsf{i}}\land \bigwedge _{\mathsf{j}=e}^{\mathsf{i}-1}\lceil \phi _{1}\rceil _{\mathsf{j}}\right) }\\
      \end{aligned}
    \end{displaymath}
          \end{definition2}
        But this is not satisfactory since the obtained schema is not sequential\footnote{Actually this is not even a schema in the sense of Definition \ref{latex_lib_label_72}
        since the upper bounds of iterations are different from $\mathsf{n}-1$.
        Notice that this is neither a regular schema \cite{tab09} since iterations
        are nested and the upper bound of one iteration contains a bound variable.}
        and, more important, because a valid LTL formula can be translated into a non-valid schema as shows the following example:
        \begin{example2}
    \begin{displaymath}
      \lceil \mathrm{X}p\Rightarrow \mathrm{F}p\rceil =p_{1}\Rightarrow \bigvee _{\mathsf{i}=0}^{\mathsf{n}}p_{\mathsf{i}}
    \end{displaymath}
                The formula $\mathrm{X}p\Rightarrow \mathrm{F}p$ is valid, but the schema $p_{1}\Rightarrow \bigvee _{\mathsf{i}=0}^{\mathsf{n}}p_{\mathsf{i}}$ is not valid (take any interpretation where $\mathsf{n}=0$).
          Adding conditions ensuring that $\mathsf{n}$ is strictly positive is possible, but obviously not sufficient,
          e.g. we could consider the formula $\mathrm{X}^{k}p\Rightarrow \mathrm{F}p$.
          Then the above translation will work only if $\mathsf{n}\geq k$ (where $k$ is arbitrary).
  \end{example2}
        The deep reason of this problem is that the semantics of schemata are intrinsically \emph{finite}
        (though unbounded) whereas those of LTL are \emph{infinite}. Actually, we can consider the previous translation as an indirect way to
        define ``finite semantics LTL'', i.e. LTL formulae interpreted over
        functions from $\{ 1, \mbox{\dots }, n\} $ to $2^{\mathcal{P}}$ for any $n\in \mathbb{N}$.
        As explained in the Introduction, LTL with finite semantics has been studied in the contexts of planning
        and runtime verification
        \cite{finite-ltl-original,ltl-finite-traces,finite-ltl-review}.
        But it seems that, rather than considering finite traces per se, the
        preferred approach in those fields is to turn finite traces into infinite
        ones by repeating infinitely the last state.
        Then the usual semantics of LTL can be used.
        Both systems seem however very similar.\subsection{A successful translation into \emph{non-SPS}}We actually need the ultimately periodic model property (Definition \ref{latex_lib_label_2} and Theorem \ref{latex_lib_label_1})
        to obtain a successful translation, written $\lceil .\rceil $, of LTL formulae into SPS.
        The aim of $\lceil .\rceil $ is that for every model ${\sigma}$ of an LTL formula ${\phi}$,
        the interpretation $\llceil \sigma \rrceil $ (Definition \ref{latex_lib_label_53}) is a model of $\lceil \phi \rceil $.
        An example is provided on Figure \ref{latex_lib_label_30}.
  \begin{figure}[h]
    \input{exspec2.tex}\caption{Specification of $\lceil .\rceil $: example.}\label{latex_lib_label_30}
  \end{figure}

  Consider an LTL formula ${\phi}$.
        As we shall see, we will make use of the schema ${\mathrm{s}}_{\leq }^{\mathrm{pfx}}$ (Proposition \ref{latex_lib_label_35})
        to enforce the fact that every model of $\lceil \phi \rceil $ is a $2$-initial segment.
        As already exposed, this $2$-initial segment is intended to denote a UP interpretation of prefix index $k$ and period $l$
        (and the parameter $\mathsf{n}$ is assigned the value $k+l-1$).
        Then the translation of ${\phi}$ (or its subformulae) will be parametrized by an arithmetic expression $e$ intended to denote the time
        (it may be either a natural number or a variable, when translating a subformula of an iteration, and it is initially equal to $0$).
        This instant will of course have an influence on the translation.
        In particular it is important to know if this instant lies in the prefix of a UP model or in its loop.
        For the prefix, we already have the propositional variable $\mathrm{pfx}$, which is specified by ${\mathrm{s}}_{\leq }^{\mathrm{pfx}}$.
        But we need to introduce a new variable for the loop, say ``$\mathrm{loop}_{e}$'', that would be true iff $e$ belongs to the loop.
        By definition, this is the case when $e\in \{ k, \mbox{\dots }, k+l-1\} $, thus we have to check that $e\geq k$ and that $e\leq \mathsf{n}$.
        By definition, the first property holds iff $\mathrm{pfx}_{e}$ does not hold.
        We thus need to express $e\leq \mathsf{n}$ with a schema, this is done as follows:\par       
  \begin{proposition2}
    \label{latex_lib_label_58}Let $( \sigma , n) $ be a schema interpretation and $e$ be a Presburger expression.
          Then $\sigma \models \langle \bigvee _{\mathsf{i}=e}^{\mathsf{n}}\top \rangle _{n}$ iff $e[n/\mathsf{n}]\leq n$.
  \end{proposition2}
  \par       
  \begin{proof}
    Indeed if $e[n/\mathsf{n}]>n$ then the iteration is empty,
          thus $\langle \bigvee _{\mathsf{i}=e}^{\mathsf{n}}\top \rangle _{n}=\bot $,
          hence cannot be satisfied by ${\sigma}$.
          Otherwise, if $e[n/\mathsf{n}]\leq n$ then the iteration is non empty,
          thus $\langle \bigvee _{\mathsf{i}=e}^{\mathsf{n}}\top \rangle _{n}$ is a non empty disjunction of ${\top}$,
          thus equivalent to ${\top}$, hence necessarily satisfied by ${\sigma}$.{\qed}
  \end{proof}
  \par       
        Thus we define $\mathrm{loop}_{e}$ as 
        follows: 
  \begin{displaymath}
    \mathrm{loop}_{e}\stackrel{{\mbox{\tiny{def}}}}{=}\neg \mathrm{pfx}_{e}\land \bigvee _{\mathsf{i}=e}^{\mathsf{n}}\top 
  \end{displaymath}
              \begin{definition2}
    \label{latex_lib_label_38}Let ${\phi}$ be an LTL formula, then $\lceil \phi \rceil $ is a schema defined as
          $\lceil \phi \rceil \stackrel{{\mbox{\tiny{def}}}}{=}\lceil \phi \rceil _{0}\land {\mathrm{s}}_{\leq }^{\mathrm{pfx}}\land \mathrm{Ax}_{\mathsf{i}=k}$\footnote{${\mathrm{s}}_{\leq }^{\mathrm{pfx}}$ is defined in Proposition \ref{latex_lib_label_35} and $\mathrm{Ax}_{\mathsf{i}=k}$ is defined before Proposition \ref{latex_lib_label_73}.}
          where, for every arithmetic expression $e$, $\lceil \phi \rceil _{e}$ is inductively defined as follows:
          \begin{displaymath}
      \begin{aligned}
        \lceil \top \rceil _{e} & \stackrel{{\mbox{\tiny{def}}}}{=}\top \\
        \lceil p\rceil _{e} & \stackrel{{\mbox{\tiny{def}}}}{=}p_{e}\\
        \lceil \neg \phi \rceil _{e} & \stackrel{{\mbox{\tiny{def}}}}{=}\neg \lceil \phi \rceil _{e}\\
        \lceil \phi _{1}\land \phi _{2}\rceil _{e} & \stackrel{{\mbox{\tiny{def}}}}{=}\lceil \phi _{1}\rceil _{e}\land \lceil \phi _{2}\rceil _{e}\\
        \lceil \mathrm{X}\phi \rceil _{e} & \stackrel{{\mbox{\tiny{def}}}}{=}( e<\mathsf{n}\land \lceil \phi \rceil _{e+1}) \lor ( e=\mathsf{n}\land {\displaystyle \bigwedge _{\mathsf{i}=0}^{\mathsf{n}}( \mathsf{i}=k\Rightarrow \lceil \phi \rceil _{\mathsf{i}}) }) \\
        \lceil \phi _{1}\mathrm{U}\phi _{2}\rceil _{e} & \stackrel{{\mbox{\tiny{def}}}}{=}\mbox{\raisebox{-0.600000cm}{$
        \begin{array}{l}
          {\displaystyle \bigvee _{\mathsf{i}=e}^{\mathsf{n}}\left( \bigwedge _{\mathsf{j}=e}^{\mathsf{i}-1}\lceil \phi _{1}\rceil _{\mathsf{j}}\land \lceil \phi _{2}\rceil _{\mathsf{i}}\right) }\lor \\
          {\displaystyle \left( \mathrm{loop}_{e}\land \bigwedge _{\mathsf{j}=e}^{\mathsf{n}}\lceil \phi _{1}\rceil _{\mathsf{j}}\land \bigvee _{\mathsf{i}=0}^{e}\left( \mathrm{loop}_{\mathsf{i}}\land \bigwedge _{\mathsf{j}=0}^{\mathsf{i}-1}( \mathrm{loop}_{\mathsf{j}}\Rightarrow \lceil \phi _{1}\rceil _{\mathsf{j}}) \land \lceil \phi _{2}\rceil _{\mathsf{i}}\right) \right) }\\
        \end{array}
        $}}\\
      \end{aligned}
    \end{displaymath}
  \end{definition2}
  \begin{example2}
    In the cases of $\mathrm{F}$ and $\mathrm{G}$, the translation simplifies drastically.
          For instance (some simple simplifications have been made):
          \begin{displaymath}
      \lceil \mathrm{F}p\rceil _{0}=\bigvee _{\mathsf{i}=0}^{\mathsf{n}}p_{\mathsf{i}}
    \end{displaymath}
          \begin{displaymath}
      \lceil \mathrm{G}p\rceil _{0}=\bigwedge _{\mathsf{i}=0}^{\mathsf{n}}p_{\mathsf{i}}
    \end{displaymath}
          This is not so simple if we consider a time $t>0$:
          \begin{displaymath}
      \lceil \mathrm{F}p\rceil _{t}=\bigvee _{\mathsf{i}=t}^{\mathsf{n}}p_{\mathsf{i}}\lor \left( \mathrm{loop}_{t}\land \bigvee _{\mathsf{i}=0}^{t}( \mathrm{loop}_{\mathsf{i}}\land p_{\mathsf{i}}) \right) 
    \end{displaymath}
          \begin{displaymath}
      \lceil \mathrm{G}p\rceil _{t}=\bigwedge _{\mathsf{i}=t}^{\mathsf{n}}\neg p_{\mathsf{i}}\land \left( \neg \mathrm{loop}_{t}\lor \bigwedge _{\mathsf{i}=0}^{t}( \mathrm{loop}_{\mathsf{i}}\Rightarrow \neg p_{\mathsf{i}}) \right) 
    \end{displaymath}
          \end{example2}

  We provide some intuitions on the transformation corresponding to the $\mathrm{X}$ and $\mathrm{U}$ connectives.
        First, for the $\mathrm{X}$:
        when computing the next instant, one has to take into account the fact that we want a UP interpretation.
        Thus if $e=\mathsf{n}$ the next time after $e$ is not $e+1$ but $k$, where $k$ is the prefix index.
        This prefix index can be specified as the only index $\mathsf{i}$ such as $\mathrm{loop}_{\mathsf{i}}\land \mathrm{loop}_{\mathsf{i}-1}$ holds.
        Second, for the $\mathrm{U}$, the first disjunct is very natural:
        it corresponds to the typical case, for instance when time $e$ occurs \emph{before} the loop.
        Then, according to the definition of the semantics of $\mathrm{U}$,
        we only have to check that $\phi _{1}$ holds on some interval $\{ e, \mbox{\dots }, \mathsf{i}-1\} $ and then that $\phi _{2}$ holds at instant $\mathsf{i}$.
        In general $\mathsf{i}$ may be arbitrary, but since the interpretation is UP,
        we can restrict to the case where $\mathsf{i}$ is in the interval $\{ e, \mbox{\dots }, k+l-1\} $, i.e. $\mathsf{i}\leq \mathsf{n}$.
              The second disjunct is slightly more complex.
        It corresponds to the case where $e$ occurs inside the periodic part of the interpretation.
        In this case, the element $\mathsf{i}$ such that $\phi _{2}$ holds may occur before $e$.
        Then $\phi _{1}\mathrm{U}\phi _{2}$ also holds if $\phi _{1}$ holds from $e$ to the end of the loop, i.e. $\mathsf{n}$,
        and then holds again when we ``get back'' at the beginning of the loop, i.e. from $k$ to some $\mathsf{i}-1$, with $\phi _{2}$ holding at $\mathsf{i}\leq e$.
        Since $\mathsf{i}\in \{ 0, \mbox{\dots }, e\} $, this can be easily stated as an iterated disjunction.
        The fact that $\mathsf{i}\geq k$ is encoded by stating that $\mathrm{loop}_{\mathsf{i}}$ must hold (i.e. $\mathsf{i}$ must be inside the periodic part of the interpretation).
  \begin{remark2}
    This transformation might remind the reader of some formulae encountered
          when dealing with the path model checking problem for UP interpretations \cite{goranko}.
          This resemblance can be explained by observing that every model of
          ${\mathrm{s}}_{\leq }^{\mathrm{pfx}_{ }}$ is a UP path, and $\lceil \phi \rceil _{e}$ is the operation of model checking the specified path.
          Then, as ${\mathrm{s}}_{\leq }^{\mathrm{pfx}_{ }}$ specifies all UP paths, we actually model check all possible models,
          hence the fact that we can conclude about the satisfiability.
  \end{remark2}

  This transformation is sound and complete but the resulting schema is not sequential 
        (iterations are nested and their bounds are different from $0$ and $\mathsf{n}-1$).
        Consequently, we present another translation in the next section, which will indeed fall in the class of SPS.\subsection{A successful translation into \emph{SPS}}The following translation follows more or less the same goal as the previous one: 
        for every model ${\sigma}$ of an LTL formula ${\phi}$, the interpretation $\llceil \sigma \rrceil $ shall be a model of $\lceil \phi \rceil $.
        Hence it relies again on the UP property.
        This new transformation uses a structure-preserving approach:
        for each subformula ${\phi}$ (different from an indexed proposition) of the original formula,
        we introduce a fresh propositional variable written $\left| \phi \right| _{ }$.
        For an indexed proposition $p$, $\left| p\right| _{ }\stackrel{{\mbox{\tiny{def}}}}{=}p$.
        Each indexed propositional variable $\left| \phi \right| _{\mathsf{i}}$, $0\leq \mathsf{i}\leq \mathsf{n}$, is then intended to be true iff the subformula ${\phi}$ is true at time $\mathsf{i}$.
        Formally, we extend $\llceil .\rrceil $ as follows:
  \begin{definition2}
    \label{latex_lib_label_61}Let ${\sigma}$ be a UP interpretation and ${\phi}$ an LTL formula. Then:
          \begin{itemize}
      \item for every propositional variable of the form $\left| \psi \right| _{ }$ for some subformula ${\psi}$ of ${\phi}$, $( \left| \psi \right| _{ }, t) \in \llceil \sigma \rrceil $ iff $\sigma , t\models \psi $; 
      \item for every other variable, $\llceil \sigma \rrceil $ is defined as described early on.
    \end{itemize}
    \par       Furthermore, for each subformula of the form $\phi _{1}\mathrm{U}\phi _{2}$,
          we add another propositional variable called $\left| \phi _{1}\mathrm{U'}\phi _{2}\right| _{ }$ 
          (called this way because its behaviour is very close to the one of $\mathrm{U}$)
          interpreted as $\mathrm{true}$ at $t\in \mathbb{N}$ iff
          there is $t'\in \mathbb{N}$ s.t. $t\leq t'\leq k+l-1$ where $\phi _{1}$ holds between $t$ and $t'-1$ and $\phi _{2}$ holds at $t'$,
          i.e. the semantics are the same as for $\mathrm{U}$ 
          except that the instant when $\phi _{2}$ occurs must happen before the end of the loop
          (as explained thereafter, this variable is used to ensure that the eventuality indeed happens).
  \end{definition2}
  Note that this semantic transformation now depends on the formula to translate. The inverse operation is defined as in Definition \ref{latex_lib_label_59}
        except that the value of any variable $\left| \psi \right| _{ }$ is ``forgotten''.
        \par 
        The translation is done by adding axioms to compute the values of the newly introduced propositional variables 
        (relating these values to the ones of the  propositional variables originally occurring in the formula).
        As we shall see, the specification of those new variables is straightforward when the head symbol of the subformula is a boolean connective:
        the value of the considered variable can be directly related to the values of the variables corresponding to the operands,
        see definition of $\mathrm{Ax}_{\neg \phi }$ and $\mathrm{Ax}_{\phi _{1}\land \phi _{2}}$ in Definition \ref{latex_lib_label_38} below.\par       When the head symbol of the subformula is a temporal connective, we have to distinguish whether the index denotes a time lower or equal to $\mathsf{n}$
        (since the interpretation is UP, we only have to consider the time interval $\{ 0, \mbox{\dots }, \mathsf{n}\} $).
        In both cases, the value of the considered propositional variable $\left| \phi \right| _{ }$ 
        at time $\mathsf{i}$ is related to the one of the variables \emph{at the next instant}.
        If $\mathsf{i}<\mathsf{n}$ then this next instant is easy to compute: it is simply $\mathsf{i}+1$.
        But if $\mathsf{i}=\mathsf{n}$, since the value of the variables $\left| \phi \right| _{ }$ are specified only on the interval 
        $\{ 0, \mbox{\dots }, \mathsf{n}\} $ we cannot refer to the time $\mathsf{n}+1$ and we have to take advantage of the fact that the interpretation is periodic:
        since $\mathsf{n}$ necessarily corresponds to the end of the periodic part, the next instant must be the beginning of the loop.
        This is easily handled in the $\mathrm{X}$ case:
        if we have $\mathrm{X}\phi $ at time $\mathsf{n}$ then we must have ${\phi}$ at time $k$ where $k$ is the beginning of the loop.\par       In the $\mathrm{U}$ case, if we have $\phi _{1}\mathrm{U}\phi _{2}$ at time $\mathsf{n}$ then we have to deal with the fact that $\phi _{2}$ might hold after $\mathsf{n}$,
        between time $k$ and $\mathsf{n}-1$ (by taking the loop into account).
        In this case we have to check that $\phi _{2}$ holds between $k$ and $\mathsf{n}-1$, and that $\phi _{1}$ holds in between.
        This check is triggered by the use of the new connective $\mathrm{U'}$,
        whose specification is thus added to the definition.
        Intuitively, $\phi _{1}\mathrm{U'}\phi _{2}$ may be seen as a connective interpreted as $\phi _{1}\mathrm{U}\phi _{2}$,
        except that the formula $\phi _{2}$ must hold \emph{at the latest} at time $\mathsf{n}$
        (one may wonder why not use directly $\mathrm{U}$ instead of $\mathrm{U'}$;
        but this would yield an ill-founded definition: 
        the eventuality could be always delayed and never fulfilled).
  \begin{definition2}
    \label{latex_lib_label_38}Let ${\phi}$ be an LTL formula. Then $\lceil \phi \rceil $ is the schema defined as
          $\lceil \phi \rceil \stackrel{{\mbox{\tiny{def}}}}{=}\left| \phi \right| _{0}\land \Phi ^{\phi }\land {\mathrm{s}}_{\leq }^{\mathrm{pfx}}\land \mathrm{Ax}_{\mathsf{i}=k}$
          where  $\Phi ^{\phi }$ stands for $\bm{\bigwedge }\{ \mathrm{Ax}_{\psi }\mid \mbox{${\psi}$ is a subformula of ${\phi}$}\} $
          and $\mathrm{Ax}_{\psi }$ is defined as follows:
          \begin{displaymath}
      \begin{aligned}
        {\mathrm{Ax}_{\top }} & \stackrel{{\mbox{\tiny{def}}}}{=}{\displaystyle \bigwedge _{\mathsf{i}=0}^{\mathsf{n}}\left| \top \right| _{\mathsf{i}}}\\
        \mathrm{Ax}_{\neg \phi } & \stackrel{{\mbox{\tiny{def}}}}{=}{\displaystyle \bigwedge _{\mathsf{i}=0}^{\mathsf{n}}( \left| \neg \phi \right| _{\mathsf{i}}\Leftrightarrow \neg \left| \phi \right| _{\mathsf{i}}) }\\
        \mathrm{Ax}_{\phi _{1}\land \phi _{2}} & \stackrel{{\mbox{\tiny{def}}}}{=}{\displaystyle \bigwedge _{\mathsf{i}=0}^{\mathsf{n}}( \left| \phi _{1}\land \phi _{2}\right| _{\mathsf{i}}\Leftrightarrow \left| \phi _{1}\right| _{\mathsf{i}}\land \left| \phi _{2}\right| _{\mathsf{i}}) }\\
        \mathrm{Ax}_{\mathrm{X}\phi } & \stackrel{{\mbox{\tiny{def}}}}{=}{\displaystyle \bigwedge _{\mathsf{i}=0}^{\mathsf{n}-1}( \left| \mathrm{X}\phi \right| _{\mathsf{i}}\Leftrightarrow \left| \phi \right| _{\mathsf{i}+1}) \land ( \left| \mathrm{X}\phi \right| _{\mathsf{n}}\Leftrightarrow \bigwedge _{\mathsf{i}=0}^{\mathsf{n}}( \mathsf{i}=k\Rightarrow \left| \phi \right| _{\mathsf{i}}) ) }\\
      \end{aligned}
    \end{displaymath}
    \begin{displaymath}
      \begin{aligned}
        \mathrm{Ax}_{\phi _{1}\mathrm{U}\phi _{2}} & \stackrel{{\mbox{\tiny{def}}}}{=}\mbox{\raisebox{-1.100000cm}{$
        \begin{array}{l}
          {\displaystyle \bigwedge _{\mathsf{i}=0}^{\mathsf{n}-1}( \left| \phi _{1}\mathrm{U}\phi _{2}\right| _{\mathsf{i}}\Leftrightarrow \left| \phi _{2}\right| _{\mathsf{i}}\lor ( \left| \phi _{1}\right| _{\mathsf{i}}\land \left| \phi _{1}\mathrm{U}\phi _{2}\right| _{\mathsf{i}+1}) ) }\\[0.200000cm]
          {\displaystyle \land ( \left| \phi _{1}\mathrm{U}\phi _{2}\right| _{\mathsf{n}}\Leftrightarrow ( \left| \phi _{2}\right| _{\mathsf{n}}\lor ( \left| \phi _{1}\right| _{\mathsf{n}}\land \bigwedge _{\mathsf{i}=0}^{\mathsf{n}}( \mathsf{i}=k\Rightarrow \left| \phi _{1}\mathrm{U'}\phi _{2}\right| _{\mathsf{i}}) ) ) ) }\\[0.200000cm]
          {\displaystyle \land \bigwedge _{\mathsf{i}=0}^{\mathsf{n}-1}( \left| \phi _{1}\mathrm{U'}\phi _{2}\right| _{\mathsf{i}}\Leftrightarrow \left| \phi _{2}\right| _{\mathsf{i}}\lor ( \left| \phi _{1}\right| _{\mathsf{i}}\land \left| \phi _{1}\mathrm{U'}\phi _{2}\right| _{\mathsf{i}+1}) ) }\\[0.200000cm]
          {\displaystyle \land ( \left| \phi _{1}\mathrm{U'}\phi _{2}\right| _{\mathsf{n}}\Leftrightarrow \left| \phi _{2}\right| _{\mathsf{n}}) }\\[0.200000cm]
        \end{array}
        $}}\\
      \end{aligned}
    \end{displaymath}
          where $\bigwedge _{\mathsf{i}=0}^{\mathsf{n}}s$ is a shortcut for $\bigwedge _{\mathsf{i}=0}^{\mathsf{n}-1}s\land s[\mathsf{n}/\mathsf{i}]$ 
          (we need to define this as an abbreviation so that the schema be indeed sequential).
          \end{definition2}
  \begin{lemma2}
    \label{latex_lib_label_60}Let ${\phi}$ be an LTL formula and $n\in \mathbb{N}$.
          The instance of $\lceil \phi \rceil $ w.r.t. $n$ contains only variables whose index is comprised between $0$ and $n$.
  \end{lemma2}
  \begin{proof}
    By inspection of all cases in Definition \ref{latex_lib_label_38},
          all indices of propositional variables in $\lceil \phi \rceil $ are $\mathsf{n}$, $\mathsf{i}$ or $\mathsf{i}+1$.
          $\mathsf{i}$ is always bound by an iteration whose bounds are $0$ and $\mathsf{n}-1$.
          Consequently the instance of $\lceil \phi \rceil $ w.r.t. $n$ only contain indexed propositional variables whose indices are between $0$ and $n$.{\qed}
  \end{proof}
  \begin{theorem}
    \label{latex_lib_label_13}Let ${\phi}$ be an LTL formula. 
          Then $\llceil .\rrceil $ is a bijection between the UP models of ${\phi}$ and the models of $\lceil \phi \rceil $
          (if the latter are restricted to the values of propositional variables occurring in the corresponding instance of $\lceil \phi \rceil $).
          $\llceil .\rrceil ^{-1}$ is the inverse bijection.
  \end{theorem}
  \begin{proof}

    We first prove that the codomain of $\llceil .\rrceil $ indeed falls into the set of models of $\lceil \phi \rceil $,
            i.e., for every UP model ${\sigma}$ of a formula ${\phi}$, $\llceil \sigma \rrceil \models \lceil \phi \rceil $.
            Let $k,l$ be the prefix index and period of ${\sigma}$.
            Note that, by Definition \ref{latex_lib_label_54}, $\llceil \sigma \rrceil $ gives the value $k+l-1$ to the parameter.
            Then we actually prove the more general result that
            for any $t\in \mathbb{N}$, if $\sigma , t\models \phi $ then $\llceil \sigma \rrceil \models {\mathrm{s}}_{\leq }^{\mathrm{pfx}}\land \left| \phi \right| _{t}\land \Phi ^{\phi }$.
            First, since $\llceil \sigma \rrceil $ is a $2$-initial segment w.r.t. $\mathrm{pfx}$, $\llceil \sigma \rrceil \models {\mathrm{s}}_{\leq }^{\mathrm{pfx}}$.
            Second, $\llceil \sigma \rrceil \models \left| \phi \right| _{t}$ by definition (Definition \ref{latex_lib_label_61}).
            Finally we prove $\llceil \sigma \rrceil \models \Phi ^{\phi }$ by proving that $\llceil \sigma \rrceil \models \mathrm{Ax}_{\psi }$ for any subformula ${\psi}$ of ${\phi}$, depending on the head symbol of ${\psi}$:
    \begin{itemize}
      \item Assume $\phi =\top $. 
                For any $t\in \mathbb{N}$, $\sigma , t\models \top $, by definition of LTL semantics.
                Thus, by definition of $\llceil \sigma \rrceil $, $\llceil \sigma \rrceil \models \left| \top \right| _{t}$ for every $t$.
                Hence, $\llceil \sigma \rrceil \models \mathrm{Ax}_{\top }$, and thus $\llceil \sigma \rrceil \models \Phi ^{\phi }$.
      \item Assume $\phi =\neg \psi $.
                For any $t\in \mathbb{N}$, $\sigma , t\models \phi $ iff $\sigma , t\not \models \psi $, by definition of LTL semantics.
                Thus, by definition of $\llceil \sigma \rrceil $, $\llceil \sigma \rrceil \models \left| \phi \right| _{t}$ iff $\llceil \sigma \rrceil \not \models \left| \psi \right| _{t}$, i.e. iff $\llceil \sigma \rrceil \models \neg \left| \psi \right| _{t}$ 
                (this time by definition of schemata semantics).
                Consequently, $\llceil \sigma \rrceil \models \left| \phi \right| _{t}\Leftrightarrow \neg \left| \psi \right| _{t}$ for any $t\in \mathbb{N}$, and in particular for $t$ between $0$ and $k+l-1$.
                Thus $\llceil \sigma \rrceil \models \mathrm{Ax}_{\phi }$.
                \item The  conjunction cases are similar.
      \item Assume that $\phi =\phi _{1}\mathrm{U}\phi _{2}$.
                Let $t\in \mathbb{N}$ be s.t. $\sigma , t\models \phi $.
                By the semantics of LTL, there exists $t'\geq t$ s.t. $\sigma , t'\models \phi _{2}$ and, for all $t''$ between $t$ and $t'-1$, $\sigma , t''\models \phi _{1}$.
                Thus either $\sigma , t\models \phi _{2}$ or $\sigma , t\models \phi _{1}$ and $\sigma , t+1\models \phi $.
                Hence, by definition of $\llceil .\rrceil $, $\llceil \sigma \rrceil \models \left| \phi _{2}\right| _{t}$ or $\llceil \sigma \rrceil \models \left| \phi _{2}\right| _{t+1}$ and $\llceil \sigma \rrceil \models \left| \phi \right| _{t+1}$,
                which enables us to conclude for the first conjunct of the $\mathrm{U}$ case.
                For the reverse implication, suppose $\llceil \sigma \rrceil \models \left| \phi _{2}\right| _{t}\lor ( \left| \phi _{1}\right| _{t}\land \left| \phi _{1}\mathrm{U}\phi _{2}\right| _{t+1}) $ for some $t\in \mathbb{N}$.
                If $\llceil \sigma \rrceil \models \left| \phi _{2}\right| _{t}$, then it is clear that $\llceil \sigma \rrceil \models \left| \phi _{1}\mathrm{U}\phi _{2}\right| _{t}$.
                If $\llceil \sigma \rrceil \models \left| \phi _{1}\right| _{t}\land \left| \phi _{1}\mathrm{U}\phi _{2}\right| _{t+1}$, then by definition $\sigma , t\models \phi _{1}\land \phi _{1}\mathrm{U}\phi _{2}$, thus $\llceil \sigma \rrceil \models \left| \phi _{1}\mathrm{U}\phi _{2}\right| _{t}$.\par           For the second conjunct 
                (i.e. the second line; notice that, for the sake of presentation simplicity, there is one conjunct per line), assume $t=k+l-1$.
                If $\sigma , t\models \phi _{2}$ then we are done.
                Otherwise, we have $t'>k+l-1$.
                This means that $\sigma , k+l\models \phi _{1}\mathrm{U}\phi _{2}$, which, by periodicity, is equivalent to $\sigma , k\models \phi _{1}\mathrm{U}\phi _{2}$, and so to $\sigma , k\models \phi _{1}\mathrm{U'}\phi _{2}$.\par           For the third and fourth conjuncts the proof is similar,
                \emph{except that we now must ensure that the instant when $\phi _{2}$ occurs must be lower or equal to $k+l-1$}.
                This is indeed the case of $t'\!\downarrow$.
                \item The case of $\mathrm{X}$ is similar (but much simpler).
    \end{itemize}

    We now focus on the inverse transformation.
            First, $\llceil .\rrceil ^{-1}$ is well defined: since $\lceil \phi \rceil $ contains ${\mathrm{s}}_{\leq }^{\mathrm{pfx}}$, every model of $\lceil \phi \rceil $ is a $2$-initial segment.
            Second, it is easily seen that $\llceil ( \llceil \mathfrak{I}\rrceil ) \rrceil ^{-1}=\mathfrak{I}$.
            Third, $\llceil ( \llceil \mathfrak{I}, n\rrceil ^{-1}) \rrceil $ is well defined since, by definition, $\llceil \mathfrak{I}, n\rrceil ^{-1}$ is UP.
            Then $\llceil ( \llceil \mathfrak{I}, n\rrceil ^{-1}) \rrceil =( \mathfrak{I}, n) $, \emph{if we restrict to the values of $\mathfrak{I}$ for indices below $n$}.
            However they might differ for indices above $n$, but, by Lemma \ref{latex_lib_label_60},
            variables with such indices do not occur in the instance of $\lceil \phi \rceil $ by $n$.
            Since we consider equality among interpretations only up to the values of propositional
            variables occurring in the corresponding instance of $\lceil \phi \rceil $,
            we indeed have the intended equality.

    We finally show that the codomain of $\llceil .\rrceil ^{-1}$ indeed falls in the set of models of ${\phi}$, i.e. that for every model $\mathfrak{M}$ of $\lceil \phi \rceil $, $\llceil \mathfrak{M}\rrceil ^{-1}\models \phi $.
            We shall prove the more general result that, for every $t\leq k+l-1$, where $k,l$ are the prefix index and the period of $\llceil \mathfrak{M}\rrceil ^{-1}$,
            and every subformula ${\psi}$ of ${\phi}$, if $\mathfrak{M}\models {\mathrm{s}}_{\leq }^{\mathrm{pfx}}\land \left| \psi \right| _{t}\land \Phi ^{\phi }$ then $\llceil \mathfrak{M}\rrceil ^{-1}, t\models \psi $.
            By induction on the structure of ${\phi}$:
    \begin{itemize}
      \item For ${\top}$  this is trivial.
      \item Assume $\phi =\neg \psi $. 
                Then $\mathfrak{M}\models {\mathrm{s}}_{\leq }^{\mathrm{pfx}}\land \left| \phi \right| _{t}\land \Phi ^{\phi }$ implies $\mathfrak{M}\models \left| \neg \psi \right| _{t}$ and $\mathfrak{M}\models \mathrm{Ax}_{\neg \psi }$, thus $\mathfrak{M}\models \neg \left| \psi \right| _{t}$, so $\mathfrak{M}\not \models \left| \psi \right| _{t}$.
                As already shown, $\mathfrak{M}=\llceil ( \llceil \mathfrak{M}\rrceil ^{-1}) \rrceil $ (as far as we consider $t\leq k+l-1$ which is indeed the case here), hence $\llceil ( \llceil \mathfrak{M}\rrceil ^{-1}) \rrceil \not \models \left| \psi \right| _{t}$.
                Hence, by contraposition w.r.t. the previous result in this proof, $\llceil \mathfrak{M}\rrceil ^{-1}, t\not \models \psi $.
                Consequently, $\llceil \mathfrak{M}\rrceil ^{-1}, t\models \phi $.
      \item For conjunction  the result is routine, using the induction hypothesis.
      \item Assume $\phi =\mathrm{X}\psi $.
                Then we have either $t<k+l-1$ or $t=k+l-1$.
                In the first case, one easily gets $\mathfrak{M}\models {\mathrm{s}}_{\leq }^{\mathrm{pfx}}\land \left| \psi \right| _{t+1}\land \Phi ^{\psi }$ by the first conjunct of $\mathrm{Ax}_{\mathrm{X}\psi }$ and concludes by induction hypothesis.
                In the second case, $\left| \phi \right| _{k+l-1}$ is equivalent to $\left| \phi \right| _{\mathsf{n}}$ (because, by definition of $\llceil .\rrceil ^{-1}$, $k+l-1$ is the value given to $\mathsf{n}$),
                so we can use the second conjunct which states that ${\psi}$ must hold at time $k$.
                By the UP property, ${\psi}$ also holds at time $k+l$, i.e. $t+1$, hence the result.
      \item Finally assume $\phi =\phi _{1}\mathrm{U}\phi _{2}$.
                We have two cases:
                Either there is some $t'$ comprised between $t$ and $k+l-1$ s.t. $\left| \phi _{2}\right| _{t'}$ holds;
                assume furthermore that $t'$ is the smallest time with this property;
                in this case, $\left| \phi _{1}\right| _{ }$ must hold between $t$ and $t'$, by the (iterated application of the) first conjunct of $\mathrm{Ax}_{\phi _{1}\mathrm{U}\phi _{2}}$;
                we then just apply the induction hypothesis to conclude.
                Or there is no such $t'$, in which case $\left| \phi _{1}\right| _{ }$ must hold from $t$ to $k+l-1$, by the same argument.
                Furthermore, since $\left| \phi _{2}\right| _{t'}$ never holds for $t'$ between $t$ and $k+l-1$, the iteration $\bigwedge _{\mathsf{i}=0}^{\mathsf{n}}( \mathsf{i}=k\Rightarrow \left| \phi _{1}\mathrm{U'}\phi _{2}\right| _{\mathsf{i}}) $ 
                of the second conjunct also holds in $\mathfrak{M}$.
                Hence $\mathfrak{M}\models \left| \phi _{1}\mathrm{U'}\phi _{2}\right| _{k}$.
                Consequently, by the iterated application of the last two conjuncts, 
                there must be some $t'$ comprised between $k$ and $k+l-1$ (actually $t-1$ is sufficient) s.t. $\left| \phi _{2}\right| _{t'}$ holds,
                and $\left| \phi _{1}\right| _{ }$ holds in between:
                indeed, the last conjunct imposes that $\left| \phi _{2}\right| _{ }$ must hold at worst at instant $\mathsf{n}$
                (note: this is precisely why $\mathrm{U'}$ is needed).
                The fact that $\left| \phi _{1}\right| _{ }$ holds in between is due to the similar structure between the two first and the two last conjuncts.
                Finally, by the UP property, the same holds for $t'+l$, which enables us to conclude.
                {\qed}
    \end{itemize}
  \end{proof}

  Furthermore it is trivial that $\#\lceil \phi \rceil $ is linear w.r.t. $\#\phi $. 
  \begin{corollary2}
    The satisfiability problem for LTL can be reduced in linear space to the satisfiability problem for SPS.
  \end{corollary2}
  \begin{theorem}
    The satisfiability problem for SPS is PSPACE-complete if numbers are encoded in unary or bounded by a constant.
  \end{theorem}
  \begin{proof}
    Consequence of the fact that the satisfiability problem for LTL is
        PSPACE-complete, of the previous corollary and of Theorem \ref{latex_lib_label_34}.{\qed}
  \end{proof}

  Notice however that this result could be proved in a much simpler way by
      directly encoding a polynomial space Turing machine with SPS.
      Such a proof would be very close to the one of Theorem 1 in \cite{planning-complexity}.

  {\bigskip}\emph{Improvements. }For practical efficiency, we can improve over Definition \ref{latex_lib_label_38}.
        We can translate the purely propositional connectives directly, i.e. without axiomatising them:
        any occurrence of an atom $\left| \top \right| _{e}$ (resp. $\left| \neg \phi \right| _{e}$, resp. $\left| \phi _{1}\land \phi _{2}\right| _{e}$)
        is directly replaced by ${\top}$ (resp. ${\neg}$$\left| \phi \right| _{e}$, resp. $\left| \phi _{1}\right| _{e}\land \left| \phi _{2}\right| _{e}$)
        repeatedly until there is no more such occurrence.
        The same applies to $\lor $, $\Rightarrow $ and $\Leftrightarrow $.
        Those are defined as abbreviations in the present paper in order to simplify definitions and proofs,
        but it is of course more efficient in practice to translate them directly when available as primitive connectives
        (obviously, this is also true for Definition \ref{latex_lib_label_8}).\par       Another optimization can be devised by observing that all schemata decision procedures \cite{tab09,nested} reason by induction on $\mathsf{n}$,
        i.e. they refute a schema for any value of $\mathsf{n}$ by reduction to the case $\mathsf{n}-1$.
        In our reduction, $\mathsf{n}$ corresponds to the last instant of the UP interpretation.
        Consequently, a schema procedure applied to a translated LTL formula
        starts by considering the \emph{last} instant of the interpretation and then going backward.
        This is counter natural since we try to refute a formula \emph{at time $0$}.
        For instance, an inductive proof is achieved for the formula $\mathrm{X}p\land \mathrm{X}\neg p$ even though this is obviously not needed:
        one would naturally try to first see what happens at time $0$ and then switch to the next state, as is done with LTL procedures.
        To tackle this problem we just need to change the translation by ``inverting the time'':
        i.e. the index $0$ will be interpreted as the last instant of the period and the index $\mathsf{n}$ as its first instant.
        Concretely, in Definition \ref{latex_lib_label_38},
        we just rewrite every index $\mathsf{i}-1$ into $\mathsf{i}$, every index $\mathsf{i}$ into $\mathsf{i}+1$, every index $0$ into $\mathsf{n}$, and every index $\mathsf{n}$ into $0$.
        Experiments with this translation indeed confirm that conjectures are refuted faster using this new translation.
      \begin{remark2}
    \label{latex_lib_label_78}The translation given here might remind the reader of bounded model checking (BMC) \cite{bmc}.
        A very important difference however is that our reduction is \emph{complete}, which is of course not the case of BMC.
        Indeed, the whole point of schemata is to reason about an infinite family of propositional formulae \emph{without having to instantiate the parameter}.
        Our translation could of course be used for BMC, simply by instantiating the parameter with successive natural numbers.
        However the converse does not hold: not every translation found in BMC could fit instead of Definition \ref{latex_lib_label_38},
        since the result must respect the syntactical criteria ensuring decidability of the satisfiability problem.
        For instance, renaming sub-formulae by propositional variables is just an optimization in the case of BMC 
        whereas in our case, it is \emph{needed} since, otherwise, the resulting schema would not be sequential (and not even regular).
        Completeness is an important problem in BMC which is usually tackled with notions like completeness thresholds and recurrence diameter \cite{bmc}
        or induction \cite{inductive-bmc}.
        Thorough analysis of how schemata procedures handle the above translation could give new ideas in order to get completeness for BMC.
  \end{remark2}
  \section{Implementation}\label{latex_lib_label_76}The implementations of both translations are available at \url{http://membres-liglab.imag.fr/aravantinos/Site/Software.html}.
      Some preliminary experiments ha{\-}ve been achieved on a few benchmarks:
      standard schemata examples provided with Reg{\-}\textsc{Stab} \cite{regstab} have been translated to LTL
      (note that the examples have been slightly modified in order to fit the constraints of SPS)
      and standard LTL pattern formulae \cite{ltl-reduction-experiments} have been translated to SPS.
      The performance of Reg{\-}\textsc{Stab} and \texttt{pltl} (\url{http://users.cecs.anu.edu.au/~rpg/software.html}) have been compared on both benchmarks.
      \emph{In both cases,} \texttt{pltl} clearly outperformed Reg{\-}\textsc{Stab}.
      We see two reasons to this:
      \begin{itemize}
    \item Reg{\-}\textsc{Stab} deals with regular schemata, which are more general than SPS.
        In particular, the decision procedure for such schemata requires the detection and elimination of pure literals
        (an adaptation of the ``Affirmative-negative rule'' of \cite{dpll}),
        which is well-known to be a huge time-consuming task 
        (and this is even more the case for schemata since we have to deal with a \emph{symbolic} notion of pure literal).
        This auxiliary procedure is needed for termination, and is mainly a consequence of the ``non-local'' aspect of schemata.
    \item 

    With LTL procedures, given a formula ${\phi}$, one knows in advance all the formulae that will occur in the deduction process:
        all of them belong to the closure of ${\phi}$ (merely the set of all subformulae of ${\phi}$, closed by negation and unfolding of temporal formulae);
        this permits the use of efficient data structures to represent sets of formulae, e.g. \texttt{pltl} uses bitsets.
        This is not the case of SPS (and even more regular schemata),
        e.g. refuting a schema containing $\bigwedge _{\mathsf{i}=0}^{\mathsf{n}}p_{\mathsf{i}}$ potentially leads to the introduction of $p_{\mathsf{n}}$, $p_{\mathsf{n}-1}$, $p_{\mathsf{n}-2}$, etc.
        By termination for regular schemata \cite{tab09}, this enumeration is finite but one does not know in advance how far it has to go.
        Hence the data structures used in Reg{\-}\textsc{Stab} are much heavier: e.g. we use balanced trees for sets of formulae.
        Thus, for big examples, the memory is easily saturated and Reg{\-}\textsc{Stab} spends much of its time in its handling
        which was absolutely not the case of \texttt{pltl}.
  \end{itemize}
      The most important reason seems to be the second one.
      It can actually be tackled in order to improve Reg{\-}\textsc{Stab} performance:
      we can syntactically extract from the input schema a bound for the above enumeration $p_{\mathsf{n}}$, $p_{\mathsf{n}-1}$, $p_{\mathsf{n}-2}$, {\dots}
      by analysis of the termination proof for regular schemata.
      Implementing this technique is ongoing work.

  Yet, there are examples where Reg{\-}\textsc{Stab} did better than \texttt{pltl}.
      Consider $( p_{1}\Rightarrow q_{\mathsf{n}+1}) \land p_{1}\land \neg q_{\mathsf{n}+1}\land \phi $ where ${\phi}$ is any formula involving some iterations.
      This schema is immediately refuted by Reg{\-}\textsc{Stab}, but the bigger ${\phi}$ is, the longer it takes for \texttt{pltl} to refute the corresponding LTL formula.
      Of course, this example was devised to emphasize one of the strengths of Reg{\-}\textsc{Stab}:
      contrarily to LTL procedures in general, and to \texttt{pltl} in particular, reasoning about schemata is \emph{global},
      i.e. Reg{\-}\textsc{Stab} may reason simultaneously on propositions containing  various \emph{symbolic} indices.
      In contrast, \texttt{pltl} will analyse the formula ${\phi}$ and the contradiction will appear only at the end of the construction 
      (i.e. by ``discovering'' eventually that $t=\mathsf{n}$ cannot hold at any state, since it would allow to derive a contradiction).\section{Discussion}\label{latex_lib_label_69}\subsection{Pros and cons of each logic}

  Since LTL and SPS are equivalent w.r.t. satisfiability, one may wonder which to favour.
        There are two major differences between LTL and schemata:
        \begin{itemize}
    \item LTL default interpretations are \emph{infinite} whereas those of schemata are \emph{finite};
    \item LTL refers to states in an \emph{anonymous} way, whereas schemata \emph{name} them.
  \end{itemize}
        These differences provide us with clear criteria for choosing one logic or the other in different situations:
        to specify an infinite behaviour, one would naturally use LTL,
        whereas classes of structurally similar finite behaviours are more naturally specified with schemata.
        Unsurprisingly, the specification of temporal behaviours falls of course in the first category.
        But, e.g., the specification of a circuit independently of the number of bits of its input falls in the second category.
        Consider for instance the specification of a ripple-carry adder:
        \begin{displaymath}
    \bigwedge _{\mathsf{i}=0}^{\mathsf{n}}( ( s_{\mathsf{i}}\Leftrightarrow ( x_{\mathsf{i}}\oplus y_{\mathsf{i}}) \oplus c_{\mathsf{i}}) \land ( c_{\mathsf{i}+1}\Leftrightarrow ( x_{\mathsf{i}}\land y_{\mathsf{i}}) \lor ( y_{\mathsf{i}}\land c_{\mathsf{i}}) \lor ( x_{\mathsf{i}}\land c_{\mathsf{i}}) ) ) \land \neg c_{0}
  \end{displaymath}
  where $x_{0}$, {\dots}, $x_{\mathsf{n}}$ and $y_{0}$, {\dots}, $y_{\mathsf{n}}$
        are the input bit vectors of size $\mathsf{n}$; $s_{0}$, {\dots}, $s_{\mathsf{n}}$ is the output bit vector and $c_{0}$, {\dots}, $c_{\mathsf{n}}$ is the carry vector. 
        Here the indices indeed correspond to the time in a \emph{concrete} sequential circuit.
        However, from a specification point of view, those indices are just an abstract way to represent a generic scheme of circuits.
        Consequently, the schema syntax seems better suited to this case (notice furthermore that it is very intuitive).\par       Similarly, the choice between a named or an anonymous representation of states depends on the situation.
        The $\mathrm{X}$ connective is well suited to express properties in a \emph{local} way, 
        since there is no need to explicitly use an index to refer to the current or the next state.
        The $\mathrm{U}$ connective is also far more intuitive than its translation to SPS
        to refer to \emph{some instant} satisfying some property in the future.
        On the other hand, in order to refer to an \emph{identified} instant of the future, one needs to refer to it by giving it a name,
        which is easily done with the schema syntax thanks to arithmetic.
        Consider e.g. the example $p_{0}\land \bigwedge _{\mathsf{i}=0}^{\mathsf{n}-1}( p_{\mathsf{i}}\Rightarrow p_{\mathsf{i}+1}) \land \neg p_{\mathsf{n}}$
        translated as $p\land \mathrm{G}( \mathrm{t<\mathsf{n}}\Rightarrow p\Rightarrow \mathrm{X}p) \land \mathrm{G}( \mathrm{t=\mathsf{n}}\Rightarrow \neg p) $
        (plus the necessary axioms ${\phi }_{<}^{\mathrm{t<\mathsf{n}}}\land \mathrm{Ax}_{\mathrm{t=\mathsf{n}}}$) in LTL.
        One can even specify behaviours \emph{after} that time (but this goes beyond sequential schemata \cite{jair2011}),
        e.g. one can write
        $p_{0}\land \bigwedge _{\mathsf{i}=0}^{\mathsf{n}}( p_{\mathsf{i}}\Rightarrow p_{\mathsf{i}+1}) \land \bigwedge _{\mathsf{i}=\mathsf{n}}^{2\mathsf{n}}( \neg p_{\mathsf{i}+1}\Rightarrow \neg p_{\mathsf{i}}) \land \neg p_{2\mathsf{n}}$.
        It seems improbable that such a property would be useful in a temporal context,
        but this could be used to specify planning problems with some predefined strategy
        e.g. if one wants to allow some set of actions in a first phase of a planning problem and then another set in some other phase of this problem.\subsection{Behaviour of $\left\lfloor .\right\rfloor $ w.r.t. LTL decision procedures}We now analyse \emph{informally} how the standard multi-pass tableau procedure of \cite{wolper} (called \textsc{LTL-tab} from now on)
        behaves on a translated schema.
        Consider the example $p_{0}\land \bigwedge _{\mathsf{i}=0}^{\mathsf{n}-1}( p_{\mathsf{i}}\Rightarrow p_{\mathsf{i}+1}) \land \neg p_{\mathsf{n}}$
        and its translation 
        $p\land \mathrm{G}( \mathrm{t<\mathsf{n}}\Rightarrow p\Rightarrow \mathrm{X}p) \land \mathrm{G}( \mathrm{t=\mathsf{n}}\Rightarrow \neg p) \land {\phi }_{<}^{\mathrm{t<\mathsf{n}}}\land \mathrm{Ax}_{\mathrm{t=\mathsf{n}}}$\footnote{Notice that this translation has been simplified since we use $\mathrm{G}( \mathrm{t=\mathsf{n}}\Rightarrow \neg p) $ instead of ${\neg}$$\mathrm{G}( \mathrm{t=\mathsf{n}}\Rightarrow p) $.}.
        We do not present a detailed tableau,
        instead we just sketch its construction by focusing on the most relevant branches 
        (the following requires some knowledge of \textsc{LTL-tab}, see \cite{wolper} otherwise).\par       When applying \textsc{LTL-tab}, the rule for the $\mathrm{U}$ connective applies on ${\phi }_{<}^{\mathrm{t<\mathsf{n}}}$ (i.e. $( \mathrm{t<\mathsf{n}}) \mathrm{U}\mathrm{G}( \neg \mathrm{t<\mathsf{n}}) $) and generates 
        one branch where $\mathrm{G}( \neg \mathrm{t<\mathsf{n}}) $ holds and one branch where $\mathrm{t<\mathsf{n}}$ and $\mathrm{X}( ( \mathrm{t<\mathsf{n}}) \mathrm{U}\mathrm{G}( \neg \mathrm{t<\mathsf{n}}) ) $ hold.
        Intuitively, the first one corresponds to the case $\mathsf{n}=0$ 
        (since it states that ${\neg}$$( \mathrm{t<\mathsf{n}}) $ always holds from the initial state till ${\omega}$)
        while the second one corresponds to $\mathsf{n}>0$ (since $\mathrm{t<\mathsf{n}}$ holds at the initial state).
        In the first case, \textsc{LTL-tab} easily finds a contradiction using mostly propositional reasoning
        (${\neg}$$\mathrm{t<\mathsf{n}}$ entails $\mathrm{t=\mathsf{n}}$ thanks to $\mathrm{Ax}_{\mathrm{t=\mathsf{n}}}$, and $\mathrm{t=\mathsf{n}}$ entails ${\neg}$$p$ with $\mathrm{G}( \mathrm{t=\mathsf{n}}\Rightarrow \neg p) $, thus yielding a contradiction).
        In the second case, since $\mathrm{t<\mathsf{n}}$ holds, one easily obtain $\mathrm{X}p$ by propositional reasoning with $\mathrm{t<\mathsf{n}}\Rightarrow p\Rightarrow \mathrm{X}p$.
        Then the decomposition of $\mathrm{Ax}_{\mathrm{t=\mathsf{n}}}$ yields $\mathrm{G}( \mathrm{t<\mathsf{n}}\land \neg \mathrm{X}( \mathrm{t<\mathsf{n}}) \Leftrightarrow \mathrm{X}( \mathrm{t=\mathsf{n}}) ) $.
        By application of the rule for $\mathrm{G}$, one immediately gets the formula $\mathrm{t<\mathsf{n}}\land \neg \mathrm{X}( \mathrm{t<\mathsf{n}}) \Leftrightarrow \mathrm{X}( \mathrm{t=\mathsf{n}}) $,
        and we then get two non-closed branches:
        one where ${\neg}$$\mathrm{X}( \mathrm{t<\mathsf{n}}) $ and $\mathrm{X}( \mathrm{t=\mathsf{n}}) $ hold (call this state ``1''),
        and one where $\mathrm{X}( \mathrm{t<\mathsf{n}}) $ and ${\neg}$$\mathrm{X}( \mathrm{t=\mathsf{n}}) $ hold (``2'').
        At the next state, we thus have two branches: one where ${\neg}$$\mathrm{t<\mathsf{n}}$ and $\mathrm{t=\mathsf{n}}$ hold, and one where $\mathrm{t<\mathsf{n}}$ and ${\neg}$$\mathrm{t=\mathsf{n}}$ hold.
        The first branch means that the instant corresponding to $\mathsf{n}$ has been reached and is easily closed similarly to the base case
        (actually, up to some formulae that only occur in the initial formula, this state is the same as the one corresponding to $\mathsf{n}=0$).
        The second branch means that $\mathsf{n}$ has still not been reached, thus we can go to the next state without encountering a contradiction.
        This is easily seen to lead either to state ``1'' or ``2'', hence the construction of the tableau terminates.
        Since ``2'' is closed the only non closed branch is the one that indefinitely loops on ``1''.
        But this loop is closed in the second pass because the eventuality $( \mathrm{t<\mathsf{n}}) \mathrm{U}\mathrm{G}( \neg \mathrm{t<\mathsf{n}}) $ is never satisfied.\par       To sum up, the construction of this tableau follows quite faithfully a proof by induction on the parameter $\mathsf{n}$.
        The axioms ${\phi }_{<}^{\mathrm{t<\mathsf{n}}}$ and $\mathrm{Ax}_{\mathrm{t=\mathsf{n}}}$ contain the arithmetic content that drive the induction,
        while $\left\lfloor s\right\rfloor _{\mathrm{prop}}$ contains the purely propositional content.
        Since LTL has to deal with infinite interpretations the induction is not well-founded in general
        (this is of course a wanted feature of LTL in order to deal with coinductive specifications).
        But the axiom ${\phi }_{<}^{\mathrm{t<\mathsf{n}}}$ introduces the eventuality $( \mathrm{t<\mathsf{n}}) \mathrm{U}\mathrm{G}( \neg \mathrm{t<\mathsf{n}}) $ which enforces a well-founded induction.
        Notice that $\left\lfloor .\right\rfloor $ can be modified so that ${\phi }_{<}^{\mathrm{t<\mathsf{n}}}$ be the only eventuality occurring in the resulting formula.
        Indeed, in its current state, the translation may introduce eventualities in two ways:
        either by negating an iteration $\bigwedge _{\mathsf{i}=0}^{\mathsf{n}-1}s$, or by negating an atom of the form $p_{\mathsf{n}+k}$.
        In the first case, the negation is equivalent to $\bigvee _{\mathsf{i}=0}^{\mathsf{n}-1}\neg s$ which can easily be simulated by the proposition $q_{\mathsf{n}}$ 
        with the axiom $\neg q_{0}\land \bigwedge _{\mathsf{i}=0}^{\mathsf{n}-1}q_{\mathsf{i}+1}\Leftrightarrow ( \neg s\lor q_{\mathsf{i}}) $.
        In the second case the translation of ${\neg}$$p_{\mathsf{n}+k}$ is ${\neg}$$\mathrm{G}( \mathrm{t=\mathsf{n}}\Rightarrow \mathrm{X}^{k}p) $.
        But, as already encountered in Example \ref{latex_lib_label_68}, this is equivalent to $\mathrm{G}( \mathrm{t=\mathsf{n}}\Rightarrow \mathrm{X}^{k}\neg p) $
        since $\mathrm{t=\mathsf{n}}$ holds at only one instant.
        Consequently one can get rid of those artificial eventualities as follows:
        \begin{itemize}
    \item put the schema in negation normal form (this introduces ${\bot}$, disjunctions and iterated disjunctions);
    \item delete every iterated disjunction by replacing it with a proposition $q_{\mathsf{n}}$ axiomatized as above;
    \item apply the translation (which is straightforwardly extended to ${\bot}$ and disjunction) by handling the case ${\neg}$$p_{\mathsf{n}+k}$ as above.
  \end{itemize}
        This is interesting since it makes the second pass much easier to handle.
        Furthermore it shows clearly that the overall proof is indeed an inductive proof,
        obtained from a coinductive proof by discarding the ill-founded branch in the second pass.\par       Proof procedures for schemata are defined by combining usual propositional procedures and inductive reasoning.
        This inductive reasoning is performed by a loop detection during the construction of the tableau.
        For instance \textsc{Stab} \cite{tab09} is defined by extending semantic tableaux.
        The reader acquainted with \textsc{Stab} may have noticed that the tableau we just sketched looks quite similar 
        to the one that would be obtained with \textsc{Stab} for the corresponding schema.
        This is mainly a matter of strategy since we oriented the construction in a way to make it understandable from a ``schema point of view''.
        There are many other tableaux that would have differed from the one obtained with \textsc{Stab}.
        The main differences between \textsc{LTL-tab} and \textsc{Stab} are the following:
        \begin{itemize}
    \item Arithmetic is handled natively in \textsc{Stab};
    \item In \textsc{LTL-tab}, termination is ensured by identifying nodes with the same labels,
            whereas this is not sufficient, in \textsc{Stab}, to ensure termination:
            a dedicated cycle relation must be defined (e.g. there is a cycle between $\bigwedge _{\mathsf{i}=0}^{\mathsf{n}-1}s$ and $\bigwedge _{\mathsf{i}=0}^{\mathsf{n}}s$).
            This is obviously not an essential difference, which is only related to the way schemata are represented and stored in the nodes;
    \item In \textsc{LTL-tab}, an artificial branch corresponding to an ill-founded derivation is discarded in the second phase,
            whereas in \textsc{Stab} the cycle relation embeds a (strict) ordering which
            ensures the well-foundedness of the derivation
            (e.g. $\bigwedge _{\mathsf{i}=0}^{\mathsf{n}}s$ cannot loop on itself).
            \emph{Consequently \textsc{Stab} does not require a second phase.}
    \item In \textsc{LTL-tab}, the reasoning is purely local, i.e. only formulae that are true at the current state are derived.
            In contrast, \textsc{Stab} may reason simultaneously on propositions containing  various \emph{symbolic} indices.
            This is related to the fact that schemata handles time in a symbolic way
            and it explains why, as mentioned in Section \ref{latex_lib_label_76},
            \texttt{pltl} performed so bad on $( p_{1}\Rightarrow q_{\mathsf{n}+1}) \land p_{1}\land \neg q_{\mathsf{n}+1}\land \phi $ where ${\phi}$ is a big formula involving some iterations.
            In contrast \textsc{LTL-tab} (and \texttt{pltl}) analyses the formula ${\phi}$ and the contradiction appears only at the end of the construction 
            (i.e. by ``discovering'' eventually that $t=\mathsf{n}$ cannot hold at any state, since it would allow to derive a contradiction).
  \end{itemize}
  \subsection{Behaviour of $\lceil .\rceil $ w.r.t. SPS decision procedures}Conversely, we can consider an LTL formula ${\phi}$ and apply \textsc{Stab} on $\lceil \phi \rceil $.
        For instance, take the unsatisfiable LTL formula $\phi \stackrel{{\mbox{\tiny{def}}}}{=}\mathrm{X}p\land \neg \mathrm{X}p$.
        The translation is then the conjunction of the following schemata 
        (we use the optimizations mentioned at the end of Section \ref{latex_lib_label_66}):
        \begin{displaymath}
    \begin{gathered}
      \left| \mathrm{X}p\right| _{\mathsf{n}}\land \left| \mathrm{X}\neg p\right| _{\mathsf{n}}\\
      \bigwedge _{\mathsf{i}=0}^{\mathsf{n}-1}( \left| \mathrm{X}p\right| _{\mathsf{i}+1}\Leftrightarrow p_{\mathsf{i}}) \\
      \left| \mathrm{X}p\right| _{0}\Leftrightarrow ( 0=k\Rightarrow p_{0}) \land \bigwedge _{\mathsf{i}=0}^{\mathsf{n}-1}( \mathsf{i}+1=k\Rightarrow p_{\mathsf{i}+1}) \\
      \bigwedge _{\mathsf{i}=0}^{\mathsf{n}-1}( \left| \mathrm{X}\neg p\right| _{\mathsf{i}+1}\Leftrightarrow \left| \mathrm{X}\neg p\right| _{\mathsf{i}}) \\
      \left| \mathrm{X}\neg p\right| _{0}\Leftrightarrow ( 0=k\Rightarrow \neg p_{0}) \land \bigwedge _{\mathsf{i}=0}^{\mathsf{n}-1}( \mathsf{i}+1=k\Rightarrow \neg p_{\mathsf{i}+1}) \\
    \end{gathered}
  \end{displaymath}
  \begin{displaymath}
    \begin{gathered}
      \neg \mathrm{pfx}_{\mathsf{n}}\Leftrightarrow n=k\\
      \bigwedge _{\mathsf{i}=0}^{\mathsf{n}-1}( \mathrm{pfx}_{\mathsf{i}}\land \neg \mathrm{pfx}_{\mathsf{i}}\Leftrightarrow \mathsf{i}=k) \\
      \neg \mathrm{pfx}_{0}\\
      \bigwedge _{\mathsf{i}=0}^{\mathsf{n}-1}( \mathrm{pfx}_{\mathsf{i}}\Rightarrow \mathrm{pfx}_{\mathsf{i}+1}) \\
    \end{gathered}
  \end{displaymath}
        It is immediately noticed that, even though the transformation is linear, the linear coefficient is very big:
        a very simple LTL formula is turned into a complicated schema.\par       We just sketch the resulting tableau.
        As explained in the previous section, the general idea of \textsc{Stab} is to refute a formula by induction on  the parameter $\mathsf{n}$.
        In the context of $\lceil .\rceil $, $\mathsf{n}$ represents the length of a UP interpretation.
        Consequently \textsc{Stab} shows that every UP interpretation falsifies the formula, by induction on $\mathsf{n}$.
        Such an approach is obviously original, but a priori not natural from an LTL point of view.
        The general scheme of the proof may be divided into three cases as follows (see Fig. \ref{latex_lib_label_74}):
        \begin{figure}[h]
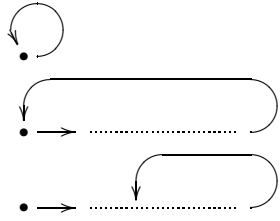

    \begin{displaymath}
      \begin{gathered}
        \mbox{\input{exinduction.tex}}\\
        \mbox{\input{exinduction2.tex}}\\
        \mbox{\input{exinduction3.tex}}\\
      \end{gathered}
    \end{displaymath}
    \caption{Proof by induction on the size of a UP interpretation}\label{latex_lib_label_74}
  \end{figure}
        either the interpretation has only one state (looping on itself), or it has more than one state.
        In the first case, there are finitely many interpretations,
        so the proof is easily achieved (simply by a tableaux-like enumeration of interpretations).
        In the second case, we encounter two more cases, depending on the position of the prefix index:
        it can either coincide with the first state, or with a state farther in the interpretation
        (formally, this corresponds to a simple case splitting on the propositional variable $\mathsf{n}=k$).
        The reasoning in each case then depends on the formula itself.\par       It is particularly interesting to understand how we deal with eventualities:
        how does it happen that we do not need a second phase, similarly to \textsc{LTL-tab}?
        To answer this question, we can observe (again informally) how \textsc{Stab} behaves on the formula $\mathrm{G}p\land \mathrm{F}\neg p$.
        Notice that we can easily define a simplified translation for the connectives $\mathrm{G}$ 
        and $\mathrm{F}$ with the following axioms:
        \begin{displaymath}
    \begin{aligned}
      \mathrm{Ax}_{\mathrm{G}\phi } & \stackrel{{\mbox{\tiny{def}}}}{=}\mbox{\raisebox{-0.400000cm}{$
      \begin{array}{l}
        \bigwedge _{\mathsf{i}=0}^{\mathsf{n}-1}( \left| \mathrm{G}\phi \right| _{\mathsf{i}}\Leftrightarrow \left| \phi \right| _{\mathsf{i}}\land \left| \mathrm{G}\phi \right| _{\mathsf{i}+1}) \\[0.200000cm]
        \land ( \left| \mathrm{G}\phi \right| _{\mathsf{n}}\Leftrightarrow \left| \phi \right| _{\mathsf{n}}\land \bigwedge _{\mathsf{i}=0}^{\mathsf{n}}( \mathsf{i}=k\Rightarrow \left| \mathrm{G}\phi \right| _{\mathsf{i}}) ) \\[0.200000cm]
      \end{array}
      $}}\\
      \mathrm{Ax}_{\mathrm{F}\phi } & \stackrel{{\mbox{\tiny{def}}}}{=}\mbox{\raisebox{-1.100000cm}{$
      \begin{array}{l}
        \bigwedge _{\mathsf{i}=0}^{\mathsf{n}-1}( \left| \mathrm{F}\phi \right| _{\mathsf{i}}\Leftrightarrow \left| \phi \right| _{\mathsf{i}}\lor \left| \mathrm{F}\phi \right| _{\mathsf{i}+1}) \\[0.200000cm]
        \land ( \left| \mathrm{F}\phi \right| _{\mathsf{n}}\Leftrightarrow \left| \phi \right| _{\mathsf{n}}\lor \bigwedge _{\mathsf{i}=0}^{\mathsf{n}}( \mathsf{i}=k\Rightarrow \left| \mathrm{F'}\phi \right| _{\mathsf{i}}) ) \\[0.200000cm]
        \land \bigwedge _{\mathsf{i}=0}^{\mathsf{n}-1}( \left| \mathrm{F'}\phi \right| _{\mathsf{i}}\Leftrightarrow \left| \phi \right| _{\mathsf{i}}\lor \left| \mathrm{F'}\phi \right| _{\mathsf{i}+1}) \\[0.200000cm]
        \land ( \left| \mathrm{F'}\phi \right| _{\mathsf{n}}\Leftrightarrow \left| \phi \right| _{\mathsf{n}}) \\[0.200000cm]
      \end{array}
      $}}\\
    \end{aligned}
  \end{displaymath}
          where $\mathrm{F'}$ is a new connective which is to $\mathrm{F}$ what $\mathrm{U'}$ is to $\mathrm{U}$.
        The case with only one state is easily handled.
        When there are more than one state, it is easily seen that the conjecture $\mathrm{G}p\land \mathrm{F}\neg p$ still holds at the next instant.
        Thus if the prefix index is above $1$, then the induction hypothesis allows to conclude immediately.
        However when the prefix is empty, the induction hypothesis does not apply:
        we actually need to make a case splitting on the value of the variable $\mathrm{F'}\neg p$:
        intuitively, this variable holds iff there is some instant before $\mathsf{n}$ s.t. $p$ holds.
        If this variable is assumed true, then we easily obtain a contradiction with $\mathrm{G}p$ (by induction).
        If it is supposed false, then we get a contradiction with the (second conjunct of the) axiom of $\mathrm{U}$ 
        which states that $\phi _{1}\mathrm{U'}\phi _{2}$ must hold at time $k$ (i.e. $0$, here), and this concludes the refutation.\par       Let us now generalize the way eventualities are handled.
        At \emph{any moment}, the procedure ``stores'' the fact that
        \emph{any} eventuality occurring as a subformula of the original formula holds or not.
        This is stored in the corresponding ``primed'' subformula (i.e. it is true iff the eventuality holds).
        Then, if ever an eventuality does not hold at the end of the period, the second conjunct of the $\mathrm{U}$ axiom imposes that
        the eventuality held before that time, inside the period.
        If this was not the case, then the corresponding primed subformula is false,
        thus we get a contradiction and this interpretation is discarded.
        On the contrary, if the eventuality held before, then we found a model.\par       The reader acquainted with the one-pass Schwendimann algorithm (``SA'') for LTL \cite{one-pass} might recognize this behaviour.
        Indeed this algorithm builds a tableau by maintaining in each state a set of \emph{unfulfilled eventualities}.
        This can be seen as corresponding to the ``primed eventualities'' of our translation.
        The set of unfulfilled eventualities at a state can be retrieved simply as the set of primed eventualites that are false at that state.
        However, apart from those informal similarities, the procedures are quite different:
        \begin{itemize}
    \item \textsc{Stab} builds explicitly a UP interpretation (that can be retrieved directly from a non closed branch of the resulting tableau)
            whereas SA just ensures that such an interpretation exists
            (which can be retrieved by loop linearization, see \cite{one-pass}, proof of Theorem 28).
            This probably makes the outcome of SA more ``understandable'', since it is more compact.
    \item In \emph{any branch}, \textsc{Stab} considers \emph{all eventualities},
            whereas SA considers only the eventualities needed for the current branch.
            This makes probably SA more efficient than \textsc{Stab} since many useless situations are trivially discarded.
    \item On the other hand, the fact that \textsc{Stab} considers all eventualities makes it possible to consider a looping in the \emph{whole tree}.
            This is not the case of SA which imposes a looping in the \emph{current branch}.
            This is precisely why the worst-case complexity of SA is bigger than the one of algorithms à la Wolper.
            Of course, an implementation of \textsc{Stab} can still impose loopings to occur only in the current branch which 
            thus makes available both possibilities to \textsc{Stab}.
            Consequently, \emph{an advantage of \textsc{Stab} is that it allows for a one-pass algorithm, while preserving an exponential time complexity}.
    \item The trade-off is that \textsc{Stab} makes some redundant computations:
            for instance, the procedure needs to ``decide'' in advance if a node is the start of the UP interpretation's loop,
            thus leading to two different branches sharing many inferences.
            With SA, the inferences are just made irrespective of whether the node will be the start of the loop or not,
            and then the loop detection is handled by the algorithm itself.
            Similarly the fact that the semantics are encoded in the translation makes \textsc{Stab}
            consider some cases that would be automatically discarded by SA.
  \end{itemize}
  \section{Model checking safety properties with schemata: an example}\label{latex_lib_label_75}With the translation given in Section \ref{latex_lib_label_66},
        and classical results of reduction from satisfiability to model checking \cite{ltl-complexity,ltl-reduction-experiments},
        one can of course use schemata to model check LTL formulae.
        However if we restrict ourselves to  n.n.f. LTL formulae whose only temporal operators are $\mathrm{X}$ and $\mathrm{G}$,
        we can obtain a much simpler translation into schemata.
        Such formulae are of interest since they can in particular model \emph{safety properties},
        i.e. formulae of the form $\mathrm{G}\psi $ where ${\psi}$ is a purely propositional formula.
        Suppose we have a transition system $T$ and want to check if it is a model of ${\phi}$.
        We first recall those notions:
  \begin{definition2}
    \label{latex_lib_label_24}A \emph{transition system} is the triple of a set of \emph{states}
          $\mathcal{S}$, a set of \emph{actions} $A$, and a transition function $\delta  :
          \mathcal{S}\times A\rightarrow \mathcal{S}$.
          A (finite or infinite) \emph{path} is a sequence of states which respects the transition function.\par       An \emph{interpreted transition system} is the pair of a transition
          system and a \emph{labelling function} $l : \mathcal{S} \rightarrow 
          2^\mathcal{P}$, where $\mathcal{P}$ is a finite set of
          propositional variables.
          As usual a \emph{computation} is a sequence of subsets of $\mathcal{P}$
          corresponding to some path of the transition system.
          For a given path ${\pi}$, we write $l( \pi ) $ for its
          corresponding computation.\par       An infinite computation can obviously be seen as an LTL interpretation
          (in the sense of Definition \ref{latex_lib_label_50}).
          Then an interpreted transition system $( T, l) $ is a \emph{model} of an LTL formula ${\phi}$ iff
          every infinite computation in $( T, l) $ is a model of ${\phi}$.
          \end{definition2}
  We now show on an example how we can model check a transition system against a formula \emph{using schemata}.
        We do not provide any formalisation since the example can easily be generalized.
        Consider the interpreted transition system $T$ represented on Figure \ref{latex_lib_label_22}.
  \begin{figure}[h]
    \begin{displaymath}
      \mbox{\input{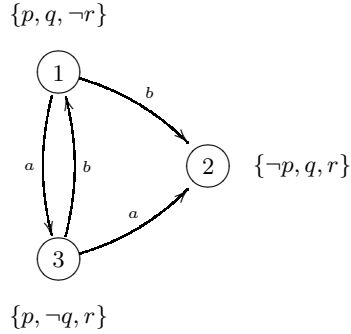}}
    \end{displaymath}
    \caption{A transition system $T$}\label{latex_lib_label_22}
  \end{figure}
  We can represent the behaviour of $T$ on all finite paths with a schema.\par       First we model the sole structure of the system, i.e. the \emph{uninterpreted} transition system.
        The indexed proposition ${\mathrm{state}}_{\mathsf{i}}^{1}$ (resp. ${\mathrm{state}}_{\mathsf{i}}^{2}$, ${\mathrm{state}}_{\mathsf{i}}^{3}$) means we are in state $1$ (resp. $2$, $3$) at time $\mathsf{i}$,
        and ${\mathrm{action}}_{\mathsf{i}}^{a}$ (resp. ${\mathrm{action}}_{\mathsf{i}}^{b}$) means that the action taken at time $\mathsf{i}$ is $a$ (resp. $b$):
        \begin{displaymath}
    \begin{array}{l}
      {\mathrm{state}}_{\mathsf{i}}^{1}\land {\mathrm{action}}_{\mathsf{i}}^{a}\Rightarrow {\mathrm{state}}_{\mathsf{i}+1}^{3}\\
      {\mathrm{state}}_{\mathsf{i}}^{1}\land {\mathrm{action}}_{\mathsf{i}}^{b}\Rightarrow {\mathrm{state}}_{\mathsf{i}+1}^{2}\\
      {\mathrm{state}}_{\mathsf{i}}^{2}\land {\mathrm{action}}_{\mathsf{i}}^{a}\Rightarrow {\mathrm{state}}_{\mathsf{i}+1}^{2}\\
      {\mathrm{state}}_{\mathsf{i}}^{2}\land {\mathrm{action}}_{\mathsf{i}}^{b}\Rightarrow {\mathrm{state}}_{\mathsf{i}+1}^{2}\\
      {\mathrm{state}}_{\mathsf{i}}^{3}\land {\mathrm{action}}_{\mathsf{i}}^{a}\Rightarrow {\mathrm{state}}_{\mathsf{i}+1}^{2}\\
      {\mathrm{state}}_{\mathsf{i}}^{3}\land {\mathrm{action}}_{\mathsf{i}}^{b}\Rightarrow {\mathrm{state}}_{\mathsf{i}+1}^{1}\\
    \end{array}
  \end{displaymath}
        Now the label of each state is easily modelled by the following schema:
        \begin{displaymath}
    \begin{array}{l}
      {\mathrm{state}}_{\mathsf{i}}^{1}\Rightarrow p_{\mathsf{i}}\land q_{\mathsf{i}}\land \neg r_{\mathsf{i}}\\
      {\mathrm{state}}_{\mathsf{i}}^{2}\Rightarrow \neg p_{\mathsf{i}}\land q_{\mathsf{i}}\land r_{\mathsf{i}}\\
      {\mathrm{state}}_{\mathsf{i}}^{3}\Rightarrow p_{\mathsf{i}}\land \neg q_{\mathsf{i}}\land r_{\mathsf{i}}\\
    \end{array}
  \end{displaymath}
        where $p_{\mathsf{i}}$ (resp. $q_{\mathsf{i}}$, $r_{\mathsf{i}}$) means that $p$ (resp. $q$, $r$) holds at time $\mathsf{i}$.
        Finally we also have to specify the fact that, at each instant $i$, there is
        one and only state active, and one and only one action can be taken
        \footnote{It is actually useless to ensure explicitly the unicity of actions since this is entailed by the unicity of states.}:
        \begin{displaymath}
    \begin{array}{l}
      {\mathrm{state}}_{\mathsf{i}}^{1}\Leftrightarrow \neg {\mathrm{state}}_{\mathsf{i}}^{2}\land \neg {\mathrm{state}}_{\mathsf{i}}^{3}\\
      {\mathrm{state}}_{\mathsf{i}}^{2}\Leftrightarrow \neg {\mathrm{state}}_{\mathsf{i}}^{1}\land \neg {\mathrm{state}}_{\mathsf{i}}^{3}\\
      {\mathrm{state}}_{\mathsf{i}}^{3}\Leftrightarrow \neg {\mathrm{state}}_{\mathsf{i}}^{1}\land \neg {\mathrm{state}}_{\mathsf{i}}^{2}\\
      {\mathrm{action}}_{\mathsf{i}}^{a}\Leftrightarrow \neg {\mathrm{action}}_{\mathsf{i}}^{b}\\
    \end{array}
  \end{displaymath}
        We write $s_{T}$ for the conjunction of all those schemata, all wrapped under a single $\bigwedge _{\mathsf{i}=0}^{\mathsf{n}}$.
        $s$ is not precisely an SPS since the upper bound of this iteration is $\mathsf{n}$ and not $\mathsf{n}-1$.
        But this is easily circumvented ($s$ is regular anyway).

  Now if we want to check this model against the formula $\mathrm{G}( p\lor q) $,
        we first translate this formula into a schema: $\lceil \mathrm{G}( p\lor q) \rceil =\bigwedge _{\mathsf{i}=0}^{\mathsf{n}}( p_{\mathsf{i}}\lor q_{\mathsf{i}}) $.
        If the transition system is indeed a model of $\mathrm{G}( p\lor q) $, then $s\Rightarrow \bigwedge _{\mathsf{i}=0}^{\mathsf{n}}( p_{\mathsf{i}}\lor q_{\mathsf{i}}) $ must be valid
        (which intuitively means that for every $n\in \mathbb{N}$ and every path of
        length $n$, the property $p\lor q$ holds all along the path).
        Equivalently, it is a model iff $s\land \bigvee _{\mathsf{i}=0}^{\mathsf{n}}( \neg p_{\mathsf{i}}\land \neg q_{\mathsf{i}}) $ is unsatisfiable
        (which means that there exists $n\in \mathbb{N}$ and a path of length $n$ s.t. the property
        $p\lor q$ does not hold at one state of the path).
        We can thus use any regular schema SAT-solver (like Reg{\-}\textsc{Stab}) to check if
        this formula is satisfiable or not.
        \section{Conclusion and future work}\label{latex_lib_label_67}LTL formulae and the so-called sequential propositional schemata have been shown to be reducible to each other in polynomial time 
      (exponential time when numbers are encoded in binary).
      This entails that the satisfiability of SPS is PSPACE-complete.
      Both those results are new.
      The reduction of SPS to LTL is not so surprising, and the converse reduction makes use
      of the well-known fact that the infinite semantics of LTL can be finitely represented.
      This remark illustrates one of the two major differences between LTL and schemata:
      whereas the semantics of LTL are infinite, those of schemata are finite.
      The other difference is that schemata allow to refer to a time in the future in a symbolic way (using the parameter $\mathsf{n}$)
      and to use arithmetic operations to construct time expressions.
      If these operations are sufficiently simple, they can be encoded in LTL formulae as shown in Section \ref{latex_lib_label_27}.
      On the other hand, LTL allows for a much handier way to deal with time in a purely \emph{local} way.\par     \emph{Future work.}
      Using the above translations to help export procedures from one logic to another is an obvious follow-up of this work
      (in particular, \textsc{Dpll} inspired procedures for schemata could help defining such a procedure for LTL).
      Similarly, as explained in Remark \ref{latex_lib_label_78}, investigating how model checking is done by translation to schemata
      could give ideas to define new completeness criteria for bounded model checking.
      The extension of the presented results to other classes of schemata could also be considered,
      e.g. schemata with nested iterations (proved decidable in \cite{nested,jair2011}).
      Translation algorithms from nested schemata into sequential ones exist \cite{jair2011}, however they are of double exponential complexity.
      Thus we conjecture that no polynomial-time transformation from nested schemata to LTL exists.  
      The extension of this study to other -- more expressive -- temporal logics could also be of interest.
      Notably, LTL with past operators \cite{glory-of-the-past} seems to be easily handled 
      with (non sequential) schemata simply by allowing negative numbers in indices.
      Since implementations for this logic do not have the same support as standard LTL and are generally not as efficient,
      such a reduction could help in improving those points.
      One could go even further by making connections between schemata and monadic second order logic (MSO).
      This would be interesting both in theory and practice, since few implementations are available for MSO 
      (only MONA \cite{mona} seems to be actively maintained).
      \par     \par     
      \bibliographystyle{amsalpha}\bibliography{ltl}

\newcommand{\etalchar}[1]{$^{#1}$}
\providecommand{\bysame}{\leavevmode\hbox to3em{\hrulefill}\thinspace}
\providecommand{\MR}{\relax\ifhmode\unskip\space\fi MR }
\providecommand{\MRhref}[2]{%
  \href{http://www.ams.org/mathscinet-getitem?mr=#1}{#2}
}
\providecommand{\href}[2]{#2}
\begin{thebibliography}{DWDMR08}

\bibitem[ACP09]{tab09}
Vincent Aravantinos, Ricardo Caferra, and Nicolas Peltier, \emph{{{A}}
  {{S}}chemata {{C}}alculus for {{P}}ropositional {{L}}ogic},
  {T}{A}{B}{L}{E}{A}{U}{X}, vol. 5607, {S}pringer, 2009, pp.~32--46.

\bibitem[ACP10a]{nested}
\bysame, \emph{{{A}} {{D}}ecidable {{C}}lass of {{N}}ested {{I}}terated
  {{S}}chemata}, in {G}iesl and {H}ähnle \cite{conf/ijcar/2010}, pp.~293--308.

\bibitem[ACP10b]{lata2010}
\bysame, \emph{{{C}}omplexity of the {{S}}atisfiability {{P}}roblem for a
  {{C}}lass of {{P}}ropositional {{S}}chemata}, {L}anguage and {A}utomata
  {T}heory and {A}pplications ({A}drian-{H}oria {D}ediu, {H}enning {F}ernau,
  and {C}arlos~{M}artín {V}ide, eds.), vol. 6031, {S}pringer, {H}eidelberg,
  2010, pp.~58--69.

\bibitem[ACP10c]{regstab}
\bysame, \emph{{{R}}eg{{S}}{{T}}{{A}}{{B}}: {{A}} {{S}}{{A}}{{T}}-{{S}}olver
  for {{P}}ropositional {{I}}terated {{S}}chemata}, in {G}iesl and {H}ähnle
  \cite{conf/ijcar/2010}, pp.~309--315.

\bibitem[ACP11]{jair2011}
\bysame, \emph{{D}ecidability and {U}ndecidability {R}esults for
  {P}ropositional {S}chemata}, Journal of Artificial Intelligence Research
  \textbf{40} (2011), 599--656.

\bibitem[AMEP10]{schema-resolution}
Vincent Aravantinos, {M}nacho {E}chenim, and Nicolas Peltier, \emph{{A}
  {R}esolution {C}alculus for {P}ropositional {S}chemata}, Tech. report, 2010,
  {A}vailable at
  {\small{\url{http://membres-lig.imag.fr/peltier/rep-AEP11.pdf}}}.

\bibitem[BCC{\etalchar{+}}03]{bmc}
Armin Biere, Alessandro Cimatti, Edmund~M. Clarke, Ofer Strichman, and Yunshan
  Zhu, \emph{Bounded model checking}, Advances in Computers \textbf{58} (2003),
  118--149.

\bibitem[BH10]{finite-ltl-review}
{A}ndreas {B}auer and {P}atrik {H}aslum, \emph{{{L}{T}{L}} {G}oal
  {S}pecifications {R}evisited}, ECAI ({A}msterdam), {I}{O}{S} {P}ress, Aug
  2010, pp.~881--886.

\bibitem[BHS98]{logic-workbench}
Peter Balsiger, Alain Heuerding, and Stefan Schwendimann, \emph{Logics
  {W}orkbench 1.0}, TABLEAUX (Harrie C.~M. de~Swart, ed.), vol. 1397, Springer,
  1998, pp.~35--37.

\bibitem[BK95]{finite-ltl-original}
{F}ahiem {B}acchus and {F}roduald {K}abanza, \emph{{U}sing {T}emporal {L}ogic
  to {C}ontrol {S}earch in a {F}orward {C}haining {P}lanner}, 3rd {E}uropean
  {W}orkshop on {P}lanning, {P}ress, 1995, pp.~141--153.

\bibitem[BM06]{ltl-finite-traces}
{J}orge~{A}. {B}aier and {S}heila~{A}. {M}cilraith, \emph{{P}lanning with
  first-order temporally extended goals using heuristic search}, National
  Conference on Artificial Intelligence, {A}{A}{A}{I} {P}ress, 2006,
  pp.~788--795.

\bibitem[{B}yl91]{planning-complexity}
{T}om {B}ylander, \emph{{C}omplexity results for planning}, {P}roceedings of
  the 12th international joint conference on {A}rtificial intelligence -
  {V}olume 1 ({S}an {F}rancisco, {C}{A}, {U}{S}{A}), {M}organ {K}aufmann
  {P}ublishers {I}nc., 1991, pp.~274--279.

\bibitem[CCG{\etalchar{+}}02]{nusmv}
Alessandro Cimatti, Edmund~M. Clarke, Enrico Giunchiglia, Fausto Giunchiglia,
  Marco Pistore, Marco Roveri, Roberto Sebastiani, and Armando Tacchella,
  \emph{{N}u{S}{M}{V} 2: {A}n {O}pen{S}ource {T}ool for {S}ymbolic {M}odel
  {C}hecking}, CAV (Ed~Brinksma and Kim~Guldstrand Larsen, eds.), vol. 2404,
  Springer, 2002, pp.~359--364.

\bibitem[CNP94]{nivat}
{H}ugues {C}albrix, {M}aurice {N}ivat, and {A}ndreas {P}odelski,
  \emph{{U}ltimately {P}eriodic {W}ords of {R}ational $\omega$-{L}anguages},
  {M}{F}{P}{S} 1994 ({L}ondon, {U}{K}), {S}pringer-{V}erlag, 1994,
  pp.~554--566.

\bibitem[DP60]{dpll}
Martin Davis and Hilary Putnam, \emph{{A} {C}omputing {P}rocedure for
  {Q}uantification {T}heory}, J. ACM \textbf{7} (1960), 201--215.

\bibitem[DWDMR08]{ltl-antichains}
M.~De~Wulf, L.~Doyen, N.~Maquet, and J.~F. Raskin, \emph{{A}ntichains:
  alternative algorithms for {L}{T}{L} satisfiability and model-checking},
  TACAS'08/ETAPS'08 (Berlin, Heidelberg), Springer-Verlag, 2008, pp.~63--77.

\bibitem[EFH{\etalchar{+}}03]{truncated-ltl}
Cindy Eisner, Dana Fisman, John Havlicek, Yoad Lustig, Anthony McIsaac, and
  David~Van Campenhout, \emph{{R}easoning with {T}emporal {L}ogic on
  {T}runcated {P}aths}, CAV (Warren A.~Hunt Jr. and Fabio Somenzi, eds.), vol.
  2725, Springer, 2003, pp.~27--39.

\bibitem[FDP01]{ltl-resolution}
Michael Fisher, Clare Dixon, and Martin Peim, \emph{Clausal temporal
  resolution}, ACM Trans. Comput. Logic \textbf{2} (2001), 12--56.

\bibitem[GH10]{conf/ijcar/2010}
{J}ürgen {G}iesl and {R}einer {H}ähnle (eds.), \emph{Ijcar}, vol. 6173,
  {S}pringer, 2010.

\bibitem[GHLS05]{lotrec}
Olivier Gasquet, Andreas Herzig, Dominique Longin, and Mohamad Sahade,
  \emph{Lo{TREC}: {L}ogical {T}ableaux {R}esearch {E}ngineering {C}ompanion},
  TABLEAUX (Bernhard Beckert, ed.), vol. 3702, Springer Berlin / Heidelberg,
  2005, pp.~318--322.

\bibitem[GPSS80]{temporal-completeness}
{D}ov {G}abbay, {A}mir {P}nueli, {S}aharon {S}helah, and {J}onathan {S}tavi,
  \emph{{O}n the temporal analysis of fairness}, POPL ({N}ew {Y}ork, {N}{Y},
  {U}{S}{A}), {A}{C}{M}, 1980, pp.~163--173.

\bibitem[HJJ{\etalchar{+}}95]{mona}
J.G. Henriksen, J.~Jensen, M.~J{\o}rgensen, N.~Klarlund, B.~Paige, T.~Rauhe,
  and A.~Sandholm, \emph{{M}ona: {M}onadic {S}econd-order logic in practice},
  TACAS '95, LNCS 1019, 1995.

\bibitem[HK03]{trp}
Ullrich Hustadt and Boris Konev, \emph{{T}{R}{P}++2.0: {A} {T}emporal
  {R}esolution {P}rover}, CADE (Franz Baader, ed.), vol. 2741, Springer, 2003,
  pp.~274--278.

\bibitem[LPZ85]{glory-of-the-past}
{O}rna {L}ichtenstein, {A}mir {P}nueli, and {L}enore~{D}. {Z}uck, \emph{{T}he
  {G}lory of the {P}ast}, CLP ({L}ondon, {U}{K}), {S}pringer-{V}erlag, 1985,
  pp.~196--218.

\bibitem[{P}nu77]{ltl}
{A}mir {P}nueli, \emph{{T}he temporal logic of programs}, {P}roceedings of
  {F}{O}{C}{S} 1977 ({W}ashington, {D}{C}, {U}{S}{A}), {I}{E}{E}{E} {C}omputer
  {S}ociety, 1977, pp.~46--57.

\bibitem[RV07]{ltl-reduction-experiments}
{K}ristin~{Y}. {R}ozier and {M}oshe~{Y}. {V}ardi, \emph{{{L}{T}{L}}
  satisfiability checking}, {P}roceedings of the 14th international
  {S}{P}{I}{N} conference on {M}odel checking software ({B}erlin,
  {H}eidelberg), {S}pringer-{V}erlag, 2007, pp.~149--167.

\bibitem[SC85]{ltl-complexity}
{A}.~{P}. {S}istla and {E}.~{M}. {C}larke, \emph{{T}he complexity of
  propositional linear temporal logics}, {J}ournal of the {A}{C}{M} \textbf{32}
  (1985), no.~3, 733--749.

\bibitem[{S}ch98]{one-pass}
{S}tefan {S}chwendimann, \emph{{A} {N}ew {O}ne-{P}ass {T}ableau {C}alculus for
  {P}{L}{T}{L}}, {T}{A}{B}{L}{E}{A}{U}{X} ({H}arrie de~{S}wart, ed.), vol.
  1397, {S}pringer {B}erlin / {H}eidelberg, 1998, pp.~277--291.

\bibitem[SSS00]{inductive-bmc}
Mary Sheeran, Satnam Singh, and Gunnar St{\aa}lmarck, \emph{Checking safety
  properties using induction and a sat-solver}, FMCAD '00, Springer-Verlag,
  2000, pp.~108--125.

\bibitem[{T}ho79]{star-free}
{W}olfgang {T}homas, \emph{{S}tar-free regular sets of $\omega$-sequences},
  {I}nformation and {C}ontrol \textbf{42} (1979), no.~2, 148 -- 156.

\bibitem[VG09]{goranko}
{V}alentin {G}oranko, \emph{{T}emporal {L}ogics for {S}pecification and
  {V}erification}, {P}roceedings of the {E}uropean {S}ummer {S}chool in
  {L}ogic, {L}anguage and {I}nformation ({E}{S}{S}{L}{I}'09), 2009.

\bibitem[{W}ol85]{wolper}
{P}ierre {W}olper, \emph{{T}he tableau method for temporal logic: an overview},
  {L}ogique et {A}nalyse \textbf{28} (1985), 119--136.

\bibitem[WVS83]{ltl-buchi}
Pierre Wolper, Moshe~Y. Vardi, and A.~Prasad Sistla, \emph{Reasoning about
  infinite computation paths}, Foundations of Computer Science, Annual IEEE
  Symposium on \textbf{0} (1983), 185--194.

\end{thebibliography}
\end{document}